\documentclass[journal]{IEEEtran}
\usepackage[font=footnotesize,labelfont=bf]{caption}
\usepackage{amsmath}
\usepackage{booktabs}
\usepackage[inline]{enumitem}
\usepackage[export]{adjustbox}
\usepackage{wrapfig}
\usepackage{subfigure}

\usepackage{amssymb,amsfonts} 
\usepackage{amsthm} 
\usepackage{commath} 
\usepackage{algorithmic}
\usepackage{mathtools, nccmath}

\usepackage{array}
\usepackage{textcomp}
\usepackage{xcolor}
\usepackage{comment}  
\usepackage[final]{pdfpages}

\usepackage{setspace}

\newcommand{\red}{\protect\textcolor{red}}

\newcommand{\dosint}[2]{%
  \ifx#1\displaystyle
    \displaysint
  \else
    \normalsint{#1}%
  \fi
}
\newcommand{\displaysint}{\displaystyle\mathsf{s}\mkern-18mu}
\newcommand{\normalsint}[1]{%
  \smallers{#1}\ifx#1\textstyle\mkern-9mu\else\mkern-8.2mu\fi
}
\newcommand{\smallers}[1]{%
  \vcenter{\hbox{$\ifx#1\textstyle\scriptstyle\else\scriptscriptstyle\fi\mathsf{s}$}}%
}

\graphicspath{{./figures/}}

\let\OLDthebibliography\thebibliography
\renewcommand\thebibliography[1]{
  \OLDthebibliography{#1}
  \setlength{\parskip}{0pt}
  \setlength{\itemsep}{0pt plus 0.3ex}
}

\begin{document}
\theoremstyle{definition}

\newtheorem{theorem}{Theorem}

\theoremstyle{definition}
\newtheorem{definition}{Definition}

\theoremstyle{definition}
\newtheorem{corollary}{Corollary}

\theoremstyle{definition}
\newtheorem{proposition}{Proposition}

\theoremstyle{definition}
\newtheorem{lemma}{Lemma}

\theoremstyle{definition}
\newtheorem{claim}{Claim}

\theoremstyle{definition}
\newtheorem{conjecture}{Conjecture}

\theoremstyle{definition}
\newtheorem{observation}{Observation}

\title {Overcoming High Frequency Limitations of \\ 
Current-Mode Control Using a Control Conditioning \\  Approach --- Part I: Modeling and Analysis}
\date{}
\author{Xiaofan~Cui,~\IEEEmembership{Student Member,~IEEE,}
        and~Al-Thaddeus~Avestruz,~\IEEEmembership{Senior Member,~IEEE}
\thanks{The authors are with the Department
of Electrical and Computer Engineering, University of Michigan, Ann Arbor,
MI, 48109 USA (e-mail: cuixf@umich.edu; avestruz@umich.edu).}
\thanks{}}

\markboth{ \hspace{1in}  {\bf MANUSCRIPT}}%
{Shell \MakeLowercase{\textit{et al.}}: Bare Demo of IEEEtran.cls for IEEE Journals}

\maketitle
\thispagestyle{headings}
\pagestyle{headings}

\renewcommand{\figurename}{Fig.}
\begin{abstract}
Current-mode control is one of the most popular controller strategies for power converters. With the advent of wide bandgap devices including GaN and SiC, higher switching frequencies have become more viable at higher power because of lower switching losses. However, the advantage of higher switching frequency for faster, higher bandwidth control is squandered because of current sensor interference. We present a framework for characterizing and analyzing this interference as uncertainties to the controller model. 
These uncertainties introduce additional dynamics and nonlinearity that can result in instability and poor transient performance of the current control loop.  In this paper, we provide a model framework based on a new control conditioning approach that guarantees global stability and a strategy for optimizing transient performance.  In Part II of this paper series, we present the analysis, design, and hardware validation of three effective solutions.
\end{abstract}

\section{Introduction}\label{sec:Intro}
\IEEEPARstart{H}{igher} switching frequency for cycle\nobreakdash-by\nobreakdash-cycle current\nobreakdash-mode control unlocks the potential of power converters~\cite{Fernandes2016}. However, peak (or valley) current-mode control frequently fails in dc-dc converters at multi-MHz switching frequencies because the frequency and amplitude of the unwanted signals in the current-sensing become comparable to that of the actual inductor current signal, consequently contaminating the current measurement and in closed loop, destabilizing the inductor current.

Current-mode control remains one of the most popular controller strategies in power management integrated circuits.  With the advent of wide bandgap devices including GaN and SiC, higher switching frequencies have become more viable at higher power because of lower switching losses \cite{Lee2013, Corradini2015}. However, the advantage of higher switching frequency for faster, higher bandwidth control is squandered because of current sensor interference.  At higher power levels, the larger physical size of the electrical circuit and components inevitably leads to greater parasitics, compromising both current sensing and the current controller.

The performance of the current-mode controller becomes hampered as the switching frequency approaches the frequency of interference because it is difficult to separate the interference from the desired current.
Traditional methods, which mainly focus on eliminating the interference, largely slows down the transient response of the high frequency power electronics.
A new approach to deal with the interference is needed to overcome this limitation when the switching frequency is high.

The voltage regulation and reference tracking using current\nobreakdash-mode control are faster than those using voltage mode because current mode results in a lower-order power converter plant \cite{erickson2007, Bao2021}.
Operating at higher switching frequency can further accelerate the transient response, especially when cycle-by-cycle control is employed \cite{Cui2018a}.
High\nobreakdash-speed current-mode control enables wide-bandwidth dc-dc power converters for emerging applications that include high-performance mobile communication systems  \cite{Cui2021apec}, mobile and embedded computing \cite{Cui2018a}, and autonomous vehicle systems \cite{Cui2019c}. Envelope tracking is important in mobile communication systems with high peak\nobreakdash-to\nobreakdash-average power ratios to improve the system efficiency in rf power amplifiers by quickly modulating the dc voltage to match the data \cite{xinbo2020, Lazarevic2019}. Controllers for high-speed voltage scaling of power converters with high granularity and wide dynamic range that can be implemented within integrated circuits for computing are particularly valuable \cite{Cui2018a}.  In autonomous vehicles, the reliability and cost-performance tradeoff of LiDAR is increased by the high-speed power scaling of laser diodes \cite{Cui2019c}.

Beyond fast transient response, high\nobreakdash-frequency current\nobreakdash-mode control also provides dc\nobreakdash-dc converters with high flexibility and reliability in distributed power systems. Current\nobreakdash-mode converters can naturally perform current sharing and are often utilized as building blocks in dc power grids \cite{Wang2018, Chang2020}, multiphase VRMs for CPUs \cite{Svikovic2015, Huang2016c,  Halivni2020}, and More Electric Aircraft \cite{Karanayil2017, Ding2020}.
The active damping characteristics of the current\nobreakdash-mode converters make it convenient to stabilize distributed systems \cite{Li2012a}, which makes them useful in telecommunications and data centers \cite{Roy2020}.
Additionally, converters that have high frequency current-mode controllers offer excellent input voltage disturbance rejection compared to those with voltage\nobreakdash-mode control \cite{erickson2007}, which is helpful in reducing input energy storage requirements in electric vehicle chargers \cite{Abedi2016} and solar converters \cite{alireza2019}.

Current\nobreakdash-mode control using current sensors \cite{Ziegler2009} offers fault tolerance to short circuits and overloads. In comparison, alternatives that include ripple\nobreakdash-based control do not offer the same protection. Compared to control methods that use inductor current estimators \cite{taeed2014compel, Chen2003tpel}, using sensors for current\nobreakdash-mode control at high frequency guarantees fast overcurrent protection \cite{Femia2020, chen2020} even in applications where inductor currents might be severely nonlinear because of inductor saturation \cite{DiCapua2016}.

Sensor interference, which is defined as the unwanted signal from a sensor, is an inescapable non-ideality in current-mode converters as shown in Fig.\,\ref{fig:noisecurrent}.
The performance of the current-mode controller becomes hampered as the switching frequency approaches the frequency of interference making it difficult to separate the inductor current ramp.

The inductor current can exhibit a severe subharmonic behavior (i.e., limit cycle) because of current sensor interference as shown in Fig.\;\ref{fig:subharmonicscurrent}. Subharmonics are illustrations of instability and can be detrimental because the dynamics of the current loop cannot be guaranteed. For example, a lower bound in frequency is difficult to guarantee, hence making it infeasible to design an output filter. Additionally, using the current loop as part of a larger control system (e.g. voltage control) becomes untenable.  We establish that this particular instability phenomena results from interference corrupting the current measurement and in the closed loop, destabilizes the inductor current. We will show in the following sections that the mechanism of this instability has a similarity to the duty\nobreakdash-cycle dependent dynamics that create a stability boundary at 50\% duty\nobreakdash-cycle in fixed-frequency current-mode converters \cite{Redl1981a}.

{\em Control conditioning} encompasses a class of methods described in this paper to mitigate the negative effects of interference on the current controller by repairing its impaired functionality to guarantee stability and optimize transient performance.
In contrast, \emph{signal conditioning} is the traditional method to handle interference by time\nobreakdash-domain or frequency\nobreakdash-domain processing of the sensor signal. For example, in current\nobreakdash-mode control, signal conditioning focuses on processing the current sensor output to closely recover the inductor current ramp; this typically results in an unacceptably conservative design and subsequent under-performance, especially when the interference and switching frequencies are near each other.

Signal blanking \cite {Li2016b}\cite{Chen2015a}, which does not enable the sampling and control until the switching transient decays, avoids interference but is only effective when the interference decays much faster than the switching period. Signal blanking not only slows down the current slew rate and transient response, but also decreases the dynamic range and causes a dead zone, which is often problematic in light\nobreakdash-load operation in peak current\nobreakdash-mode converters and heavy\nobreakdash-load operation in valley current\nobreakdash-mode converters. Additionally, overcurrent protection cannot be performed during blanking. 
Aggressive low\nobreakdash-pass filtering \cite{Simon-Muela2008} is an alternative when signal blanking is inadequate. For this to be effective, the interference frequency needs to be much higher than the switching frequency so that the filter can attenuate the interference without significantly distorting the inductor current ramp. Under\nobreakdash-performance often occurs because the settling time of the voltage\nobreakdash-control loop is bounded from above by the settling time of the filter.  Dampers or snubbers can be used to reduce the ringing from the parasitics, but can be detrimental to the performance of the power circuit in different ways, including efficiency.


We investigate a \emph{control conditioning} approach, which can overcome the performance limitations as higher switching frequencies approach the frequency bands of interference.  Control conditioning is a control\nobreakdash-performance\nobreakdash-oriented design which is superior to signal conditioning in a number of applications. A classic example of control conditioning is slope compensation \cite{Chen2014d}, which is aimed at stabilizing fixed-frequency, peak current-mode control when the duty cycle exceeds 50\%.

In this paper, control conditioning relies on two current mappings --- static current mapping and dynamical current mapping, which can be corrupted by interference, as discussed in Section\,\ref{sec:modeling}.
The static current mapping can be extended to quasi\nobreakdash-static current mapping if the interference is not synchronized with switching. 
Section\,\ref{sec:conclusion_cmc_part1} concludes the paper.
In Part II of this paper series, we use the control conditioning framework to repair these current mappings using three methods. We are able to rigorously stabilize the current controller while optimizing for transient response, which is a key metric for the {\em control performance} of power electronics. 

\begin{figure}
    \centering
    \includegraphics[width=8cm]{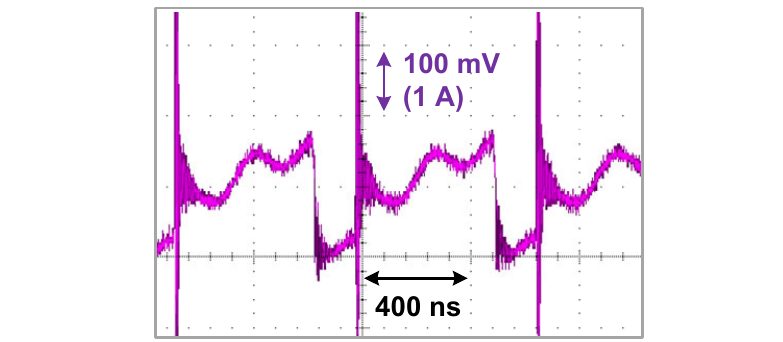}
    \caption{Current sense voltage of a current-mode boost converter using constant off-time (100\,m$\Omega$ current sense resistor). The current sensor output is largely distorted by interference, and the measurement error can be as much as $50\%$. The interference comes from the parasitic ringing and power ground resonance.}\label{fig:noisecurrent}
\end{figure}

\begin{figure}
    \centering
    \includegraphics[width=8cm]{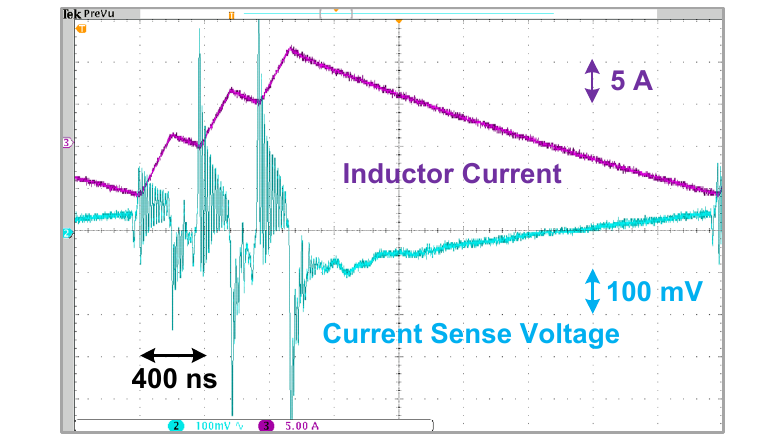}
    \caption{Subharmonics on the inductor current waveform of a current-mode buck converter using constant on-time because the interference severely destabilizes the current controller.}\label{fig:subharmonicscurrent}
\end{figure}

\begin{figure}
    \centering
    \includegraphics[width=8 cm]{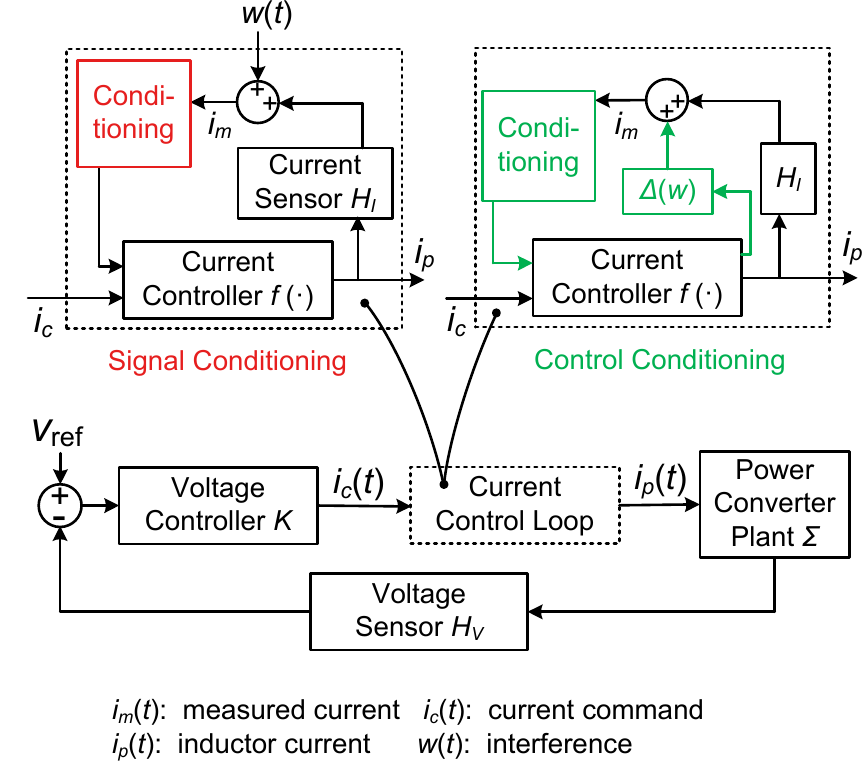} 
    \caption{Current-mode control includes a voltage control loop and a current control loop. The current control loop consists of a current sensor $H_I$ and a nonlinear current controller $f(\cdot)$ which is represented by a state machine in Fig.\, \ref{fig:dynmap_b}. The interference contaminates the current sensor output $i_m(t)$ and degrades the performance of the current control loop.}
    \label{fig:dynmap}
\end{figure}

\begin{figure}
    \centering
    \includegraphics[width= 7cm]{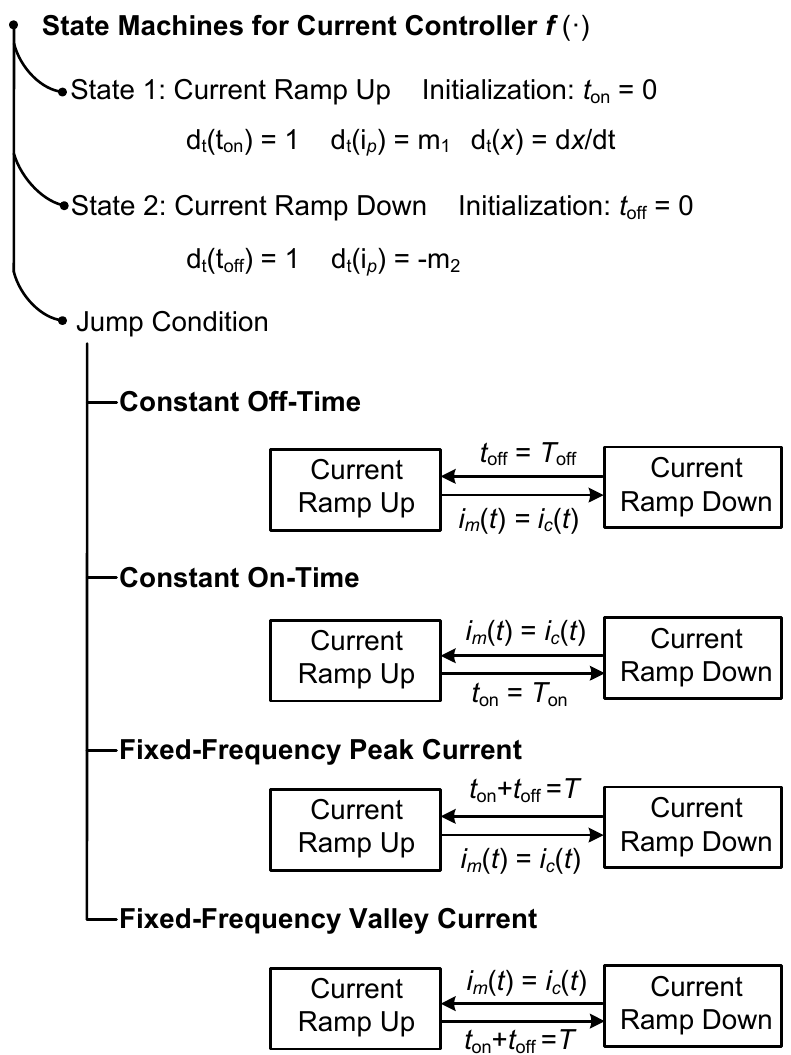} 
    \caption{The nonlinear current controller $f(\cdot)$ can be represented by a state machine.}
    \label{fig:dynmap_b}
\end{figure}
\section{Control Conditioning Approach} \label{sec:modeling}
{\em Control conditioning} overcomes a critical longstanding limitation in switching frequency for peak current mode control.  This approach can outperform {\em traditional signal conditioning} by embedding the interference into the control loop rather than as an input, hence providing a framework for directly optimizing control performance.
For a power converter where the output voltage is controlled, the current controller is a minor feedback loop within the voltage loop \cite{kuo1987automatic}. We refer to this minor loop as the {\em current control loop} as shown in Fig.\,\ref{fig:dynmap}.
The current control loop used for current-mode control consists of a current sensing $H_I$ and nonlinear controller $f(\cdot)$.
The nonlinear controller $f(\cdot)$ can be represented by a state machine \cite{Liberzon2003a} as shown in Fig.\, \ref{fig:dynmap_b}.

In the traditional signal conditioning approach, tools are directly applied to repair the corrupted time-domain current sensing signal $i_m(t)$ so that the signal more closely appears like the ideal, with the expectation that this will improve control performance.  Although this approach is straightforward, this does not necessarily mean that the control performance improves.
\begin{figure}
    \centering
    \includegraphics[width = 6cm]{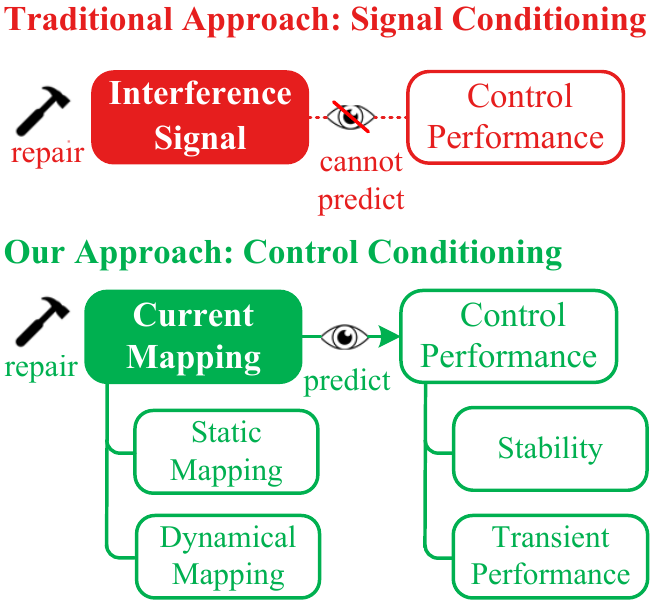}
\caption{Control conditioning strives to repair the current mappings whereas signal conditioning attempts to process the current sensor output. Control conditioning directly dictates control performance through the current mapping, which is in contrast to signal conditioning methods where the control performance cannot be directly designed.}
    \label{fig:ctrlcond_vs_sigcond}
\end{figure}

Control conditioning models the effect of interference on a peak current sensor as uncertainties $\Delta(w)$ in the current control loop. The current control loop maps the command current $i_c(t)$ to the actual peak inductor current $i_p(t)$ as shown in Fig.\,\ref{fig:dynmap}. This enables the use of more powerful tools within control theory to attain better performance.
Instead of directly repairing the time-domain current-sense waveform, the control conditioning approach uses tools to repair the degraded current control loop. The {\em static current mapping} and {\em dynamical current mapping} are the two related mappings that determine the behavior of the current control loop when the loop is affected by interference. A comparison between the control conditioning approach to traditional signal conditioning is summarized in Fig.\,\ref{fig:ctrlcond_vs_sigcond}.

The methods that are discussed in this paper often refer to peak current mode control with its corresponding constant off time, but these methods also directly apply to valley current mode control and its corresponding constant on time.  Additionally, the analysis and control conditioning methods discussed apply to fixed frequency converters as well. In the following several sections, this is discussed.

\subsection{Interference} \label{sec:ctrlcond_intf}
Interference can result from mechanisms that are either or both endogenous, such as ringing and nonlinearities from parasitics in power and sensing circuits, or exogenous, such as magnetic and electric field coupling. 
Parasitics that include the equivalent series inductance of the output capacitor, self\nobreakdash-resonance of the power inductor, and the output capacitance and lead inductance of the semiconductor switches can be sources of interference.
Current sensors with its parasitic inductance directly in the current path can contribute to ringing. Non-contact sensors like current transformers can be troublesome as well because of self-resonance \cite{McNeill2004}.  Other non-contact varieties (e.g. Hall or GMR sensors \cite{Singh2008a}) can also pick up interference.

Unlike strict circuit noise \cite{lee2003design}, which is purely an unbounded stochastic process, interference can be modeled as a deterministic bound.  For example, the waveform for the interference in Fig.\,\ref{fig:noisecurrent} has an underlying process (e.g. resonant ringing) that is deterministic, or in other cases can be random, but with a bound on the amplitude and bandwidth.  This model for interference permits using analysis tools for deterministic systems, making the control problem tractable.  Analysis tools that are deterministic use bounds on uncertainty, which can be considered {\bf\it worst-case analysis}.

The effect of interference on a peak current controller depends on whether the trigger instance for the peak current command is correct at each individual cycle.
The error from interference is characterized by an error $\Delta t$ in the trigger instance, which can be either early or late.
This trigger instance error only depends on the deviation $\Delta x$ of the actual signal from the ideal signal as shown in Fig.\,\ref{fig:Additive}.  The control conditioning framework outlined in this paper applies to a wide range of interference that is encountered.
\begin{figure}
    \centering
    \includegraphics{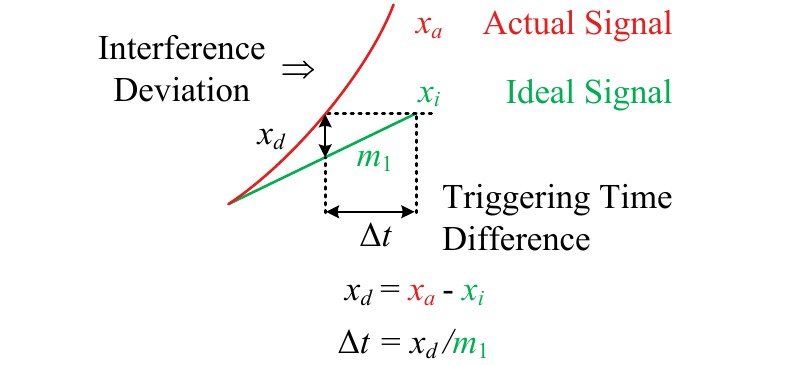}
    \caption{Each independent trigger instance deviates from the ideal signal because of interference.  This deviation can be modeled as additive.}
    \label{fig:Additive}
\end{figure}

\subsection{Static Current Mapping} \label{sec:static_map}
The static current mapping is a function from the current command $i_c$ to the actual peak inductor current $i_p$ in the periodic steady state as shown in Fig.\,\ref{fig:inf_statmap}.
A broad class of interference can be modeled by a static current mapping. 

In the static current mapping, $i_p = i_c$ if the current sensing is ideal. Given different non-idealities in the current sensing, the static current mapping can be deformed differently as illustrated in Fig.\,\ref{fig:inf_statmap}. For example, a sensor gain error can be expressed as an overall gain error for the current control loop. This does not affect the stability nor performance of the current control loop, but rather it affects the encompassing voltage control loop and can be repaired with another gain block.  An offset error in the current sensor appears as an overall offset error for the current control loop;  this offset error appears as a disturbance to the voltage control loop that is typically corrected by the inclusion of an integral controller \cite{erickson2007}.
\begin{figure}
    \centering
    \includegraphics[width = 8cm]{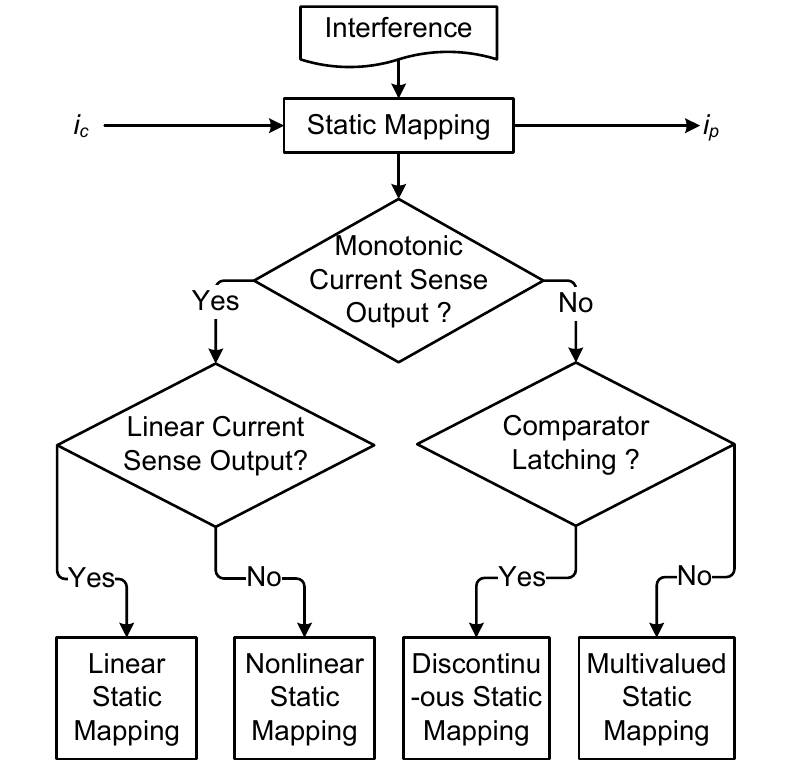}
    \caption{Depending on the type of interference on current sensor, the linear static current mapping can be deformed into nonlinear, discontinuous or multivalued mapping.} \label{fig:inf_statmap}
\end{figure}

Interference corrupts this ideal static mapping by creating discontinuities and other nonlinearities.
The most prevalent nonlinearity that arises from current-peak detection based on a hardware comparator, shown in Fig.\,\ref{fig:multieventtrigger}, is multivaluedness.
Because interference can cause the inductor current ramp signal to cross the comparator threshold multiple times within a switching cycle, the actual peak current can vary from cycle to cycle despite a non-varying command because the event latching depends on propagation delays and the digital clock as shown in Fig.\,\ref{fig:multieventtrigger}. This multivalued peak current can result in large deviations in the output of the outer voltage control loop.  Additionally, the stability of the outer voltage loop requires the static map to be monotonic.

To use the control conditioning methods for dynamical mappings, the static current mapping can be nonlinear, but must be {\it{proper}}:
\begin{enumerate*}
\item  $\alpha$-sector-bounded;
\item single-valued;
\item smooth;
\item monotonic.
\end{enumerate*}
\begin{definition}
    A function $h\,:\mathbb{N}\,\times \mathbb{R} \rightarrow \mathbb{R}$ is said to be \emph{$\alpha$\nobreakdash-sector\nobreakdash-bounded} to $[K_{lb}, K_{ub}]$ if there exists a constant $\alpha$ such that
    \begin{align}
        \left[ h(n,x) - K_{lb} x - \alpha \right]\left[ h(n,x) - K_{ub} x - \alpha \right] \le 0.
    \end{align}
\end{definition}
The static current mapping is $\alpha$-sector-bounded if the interference is bandwidth-limited and amplitude-limited as shown in Appendix \ref{ref:sec_boun_Tmap}.  The degree of nonlinearity is defined as the largest absolute fractional deviation of the relevant static mapping from the ideal static mapping.

The control conditioning strategy that we introduce for the static current map is to first repair the multivaluedness so that it becomes single-valued, but possibly discontinuous, as shown in Fig.\,\ref{fig:w_o_latch}(a). We then use another control conditioning method to smooth the discontinuity, as shown in Fig.\,\ref{fig:w_o_latch}(b). A discontinuity is problematic in that it can spawn limit cycles in the outer voltage control loop.

The multivalued mapping from a hardware comparator threshold detector can be repaired so that it is single-valued by using {\em first-event triggering with latching}.

\subsection{First-Event Triggering with Latching}  \label{sec:fet_latching}
The traditional design of peak current-mode control typically does not carefully determine the triggering criterion. For example, not definitively determining an event when there are multiple crossings within a switching cycle; in Fig.\,\ref{fig:multieventtrigger}, the static mapping is multivalued, in other words indeterminate. The choice of triggering criterion is important because it can significantly affect the stability of the current control loop.
\begin{figure}
    \centering
    \includegraphics[width = 6cm]{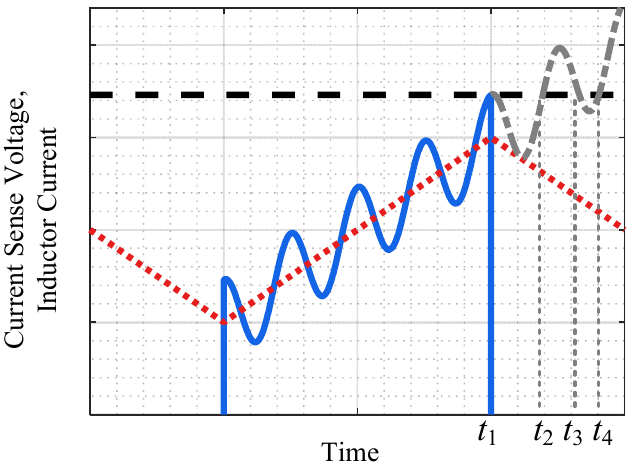}
    \caption{First-event triggering with latching using a comparator prevents the multiple instances when the current sensor waveform crosses the trigger from causing an indeterminate event instance.
    ({\color{blue}---}) and ({\color{gray}-\,$\cdot$\,-}) represent the current-sense voltage; ({\color{red}$\cdots$}) represents the inductor current; ({\color{black}-\,-\,-}) represents the current command.}
      \label{fig:multieventtrigger} 
\end{figure}
\begin{figure}
    \centering
    \includegraphics[width = 5cm]{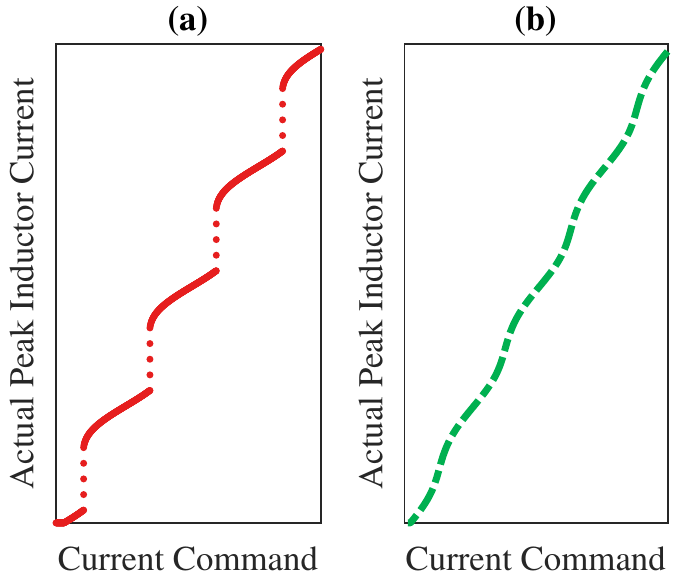}
    \caption{Static mapping using first-event triggering with latching results in: (a)\,a discontinuous and nonlinear, but single-valued static mapping, which can be transformed with further control conditioning to: (b)\,a smooth and nonlinear static mapping.}
    \label{fig:w_o_latch}
\end{figure}

We introduce a particular triggering criterion as part of control conditioning ---  first-event triggering with latching.  We show that this criterion makes a multivalued static mapping single-valued, although possibly discontinuous as illustrated in Fig.\,\ref{fig:w_o_latch}(a). First-event triggering with latching is not the only method to solve the problem of multivalued static mapping.
More sophisticated triggering criteria are possible using information about the structure of the interference, e.g. event-triggering at the $n^{\text{th}}$ threshold-crossing or using higher complexity decisions. We use first-event triggering with latching because it is the most straightforward event-triggering mechanism.

First-event triggering with latching can be implemented using a comparator to detect a threshold crossing and a flip\nobreakdash-flop to latch this first event as shown in Fig.\,\,\ref{fig:eboostschmatic}.
For example, in Fig.\,\ref{fig:multieventtrigger}, the comparator output changes state at $t_1$, $t_2$, $t_3$, and $t_4$.  The D flip\nobreakdash-flop always latches the first rising edge of the comparator; this ensures that only the edge at $t_1$ triggers the event.
The D flip\nobreakdash-flop is then reset by the digital controller at the beginning of the next switching cycle. The resulting static mapping, as shown in Fig.\,\ref{fig:w_o_latch}(a), is a single\nobreakdash-valued function.

For a fixed\nobreakdash-frequency peak current\nobreakdash-mode converter, first\nobreakdash-event\nobreakdash-triggering with latching is inherently implemented in a typical controller as shown in Fig.\,\ref{fig:fixedfreqpccboost}.  The RS flip\nobreakdash-flop always latches the first level\nobreakdash-change of the comparator output.  The RS flip\nobreakdash-flop is then reset by the fixed\nobreakdash-frequency clock at the next switching cycle.

Proposition \ref{th:equexcond} articulates the relationship between the static mapping and current sensor's time\nobreakdash-domain output signal.
The static current mapping definitively determines the actual current from the current command, 
whereas the current sensor output evolves with time.  
Although these two functions belong to different domains, they are related by the triggering criterion.  
Our criterion of first\nobreakdash-event\nobreakdash-triggering with latching preserves an important analytical property between the static mapping and current sensor output function.  The proof is found in Appendix \ref{proof:equexcond}. 
\begin{proposition} \label{th:equexcond}
For a current control loop using either constant off\nobreakdash-time current control or fixed-frequency peak current\nobreakdash-mode control, along with first\nobreakdash-event\nobreakdash-triggering with latching,
the static mapping $\mathcal{T}$ is a strictly monotonically increasing and continuous mapping if and only if the current sensor output function \mbox{$m_1t+w(t)$} is strictly monotonically increasing and continuous, where $m_1$ is the rising slope of the inductor current ramp and $w(t)$ is the interference signal.
\end{proposition}

Proposition 2 is analogous to Proposition 1 for
the current control loop of constant on\nobreakdash-time current-mode control and fixed-frequency valley current-mode control and can be similarly proven.
\begin{proposition} \label{th:equexcond_ff}
For a current control loop using either constant on\nobreakdash-time current control or fixed-frequency valley current\nobreakdash-mode control, along with first\nobreakdash-event\nobreakdash-triggering with latching,
the static mapping $\mathcal{T}$ is a strictly monotonically increasing and continuous mapping if and only if the current sensor output \mbox{$-m_2t+w(t)$} is strictly monotonically increasing and continuous, where $m_2$ is the falling slope of the inductor current ramp and $w(t)$ is the interference signal.
\end{proposition}

Guaranteeing stability using the dynamical mapping in Section\,\ref{sec:dyn_map} requires the static mapping to be monotonic and continuous.  Proposition\,\ref{th:equexcond} and Proposition\,\ref{th:equexcond_ff} indicate that the continuity and monotonicity of the static mapping is equivalent to the continuity and monotonicity of the current sensor output signal. This correspondence between the time-domain signal and the static mapping means that within the dynamical mapping, the dynamical system which operates on the output signal can be used to ensure continuity and monotonicity of the static mapping.

\subsection{Quasi\nobreakdash-Static Current Mapping} \label{sec:qs_mapping}
A quasi-static (QS) current mapping $\breve{T}$ is related to the static current mapping $\mathcal{T}$. The QS current mapping is a family of cycle\nobreakdash-by\nobreakdash-cycle functions from the current command $i_c$ to the actual peak inductor current $i_p$ as shown in Fig.\,\ref{fig:qsmapping}, and applies to a class of interference that does not repeat every switching period.  Fortunately, the guarantee of stability of the dynamical model discussed in Section\,\ref{sec:dyn_map} does not require the QS mapping to be static; it only requires that the QS current mapping be sector\nobreakdash-bounded and the members satisfy the same properties of a proper static map.
\begin{figure}[ht]
\centering
\subfigure[Schematic]{
\includegraphics[width= 8 cm]{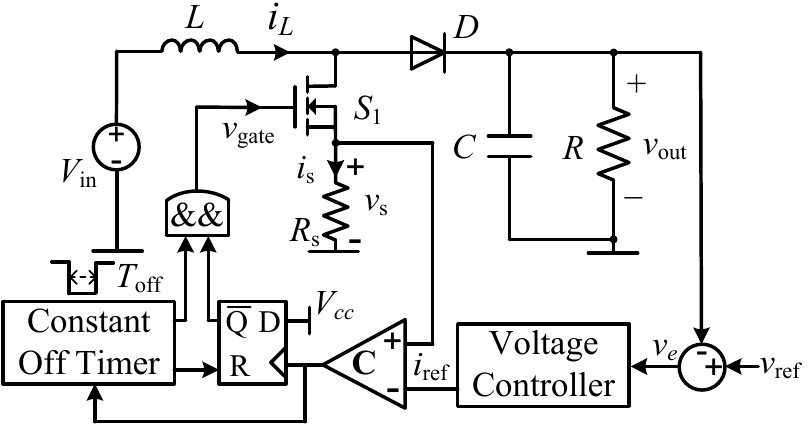} 
\label{fig:eboostschmatic}
}
\subfigure[Inductor current and capacitor voltage]{
    \includegraphics[width=8cm]{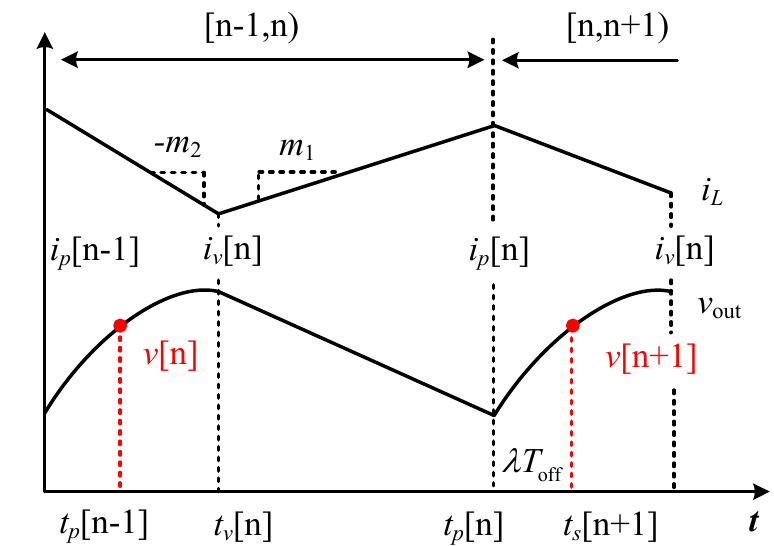}
    \label{fig:eboostwaveform}
}
\caption{Current\nobreakdash-mode boost converter using constant off\nobreakdash-time requires a reference frame for the non\nobreakdash-uniform sampling of peak current and output voltage.}
\end{figure}
\begin{figure}[ht]
\centering
\begin{subfigure}[Schematic] 
    {\centering
    \includegraphics[width = 8cm]{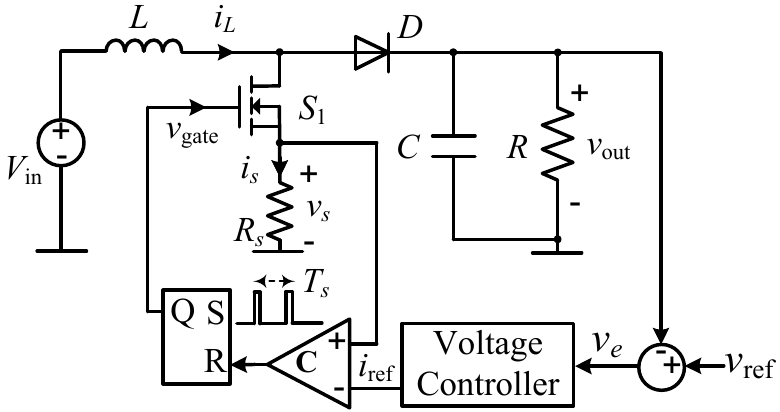}
    \label{fig:fixedfreqpccboost}
    }
\end{subfigure}
\newline
\begin{subfigure}[Capacitor voltage and inductor current]
    {
    \centering
    \includegraphics[width = 8cm]{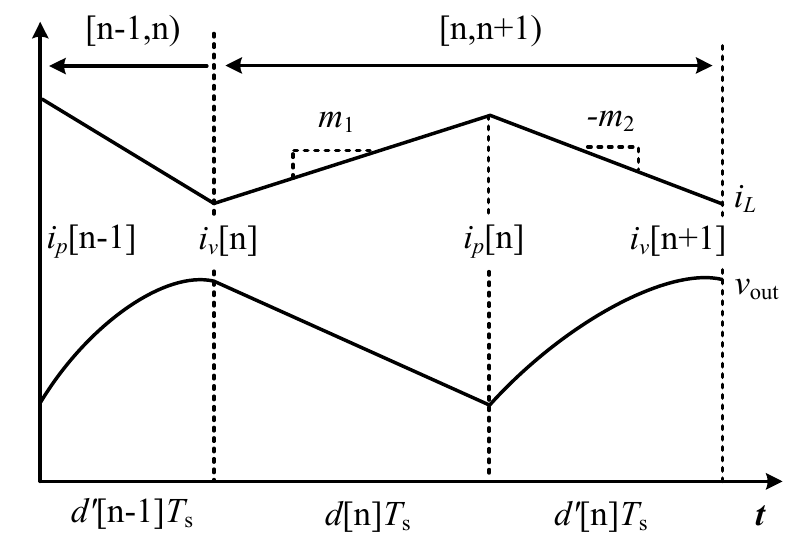}
    \label{fig:fpccwaveform}
    }
\end{subfigure}
\caption{Fixed\nobreakdash-frequency current\nobreakdash-mode boost converter requires non\nobreakdash-uniform sampling for the peak current even though the output is uniformly sampled.
}
\end{figure}

QS mapping admits a class of interference, whose frequency and amplitude are bounded, but does not repeat every cycle.  This can exist from exogenous and endogenous sources.
As discussed in Section\,\ref{sec:ctrlcond_intf}, the interference waveform does not have to repeat every switching cycle to model the interference in the control conditioning framework.
The non\nobreakdash-repeating interference has two pertinent effects: (i) an uncertainty in the current control loop and (ii) a feedthrough to the peak current output, as shown in Fig.\,\ref{fig:interference_feedforward}. In Section\,\ref{sec:dyn_map}, we show how this uncertainty (represented by the QS mapping) becomes a time\nobreakdash-varying nonlinearity in the dynamical mapping, which can be stabilized by proper design. This precludes undesirable subharmonic instabilities, for example in fixed\nobreakdash-frequency converters operating in peak current\nobreakdash-mode control \cite{Redl1981a}. Because the QS mapping is $\alpha$\nobreakdash-sector\nobreakdash-bounded and the dynamics of the control conditioning are designed to minimize the variations as a result of the interference, the inductor current perturbations from the interference feedthrough is bounded in amplitude and frequency. The peak current variations only appear as disturbances in the voltage control loop of a power converter. High frequency variations in the voltage output can in practice be reduced by additional filtering capacitance. Lower frequency variations are rejected by the voltage control loop.

The goal of the quasi\nobreakdash-static mapping and its usage in the dynamical mapping is to ensure the stability of the current control loop and to present a stable plant to the voltage control loop. In Section\,\ref{sec:dyn_map}, we show that the stability of the dynamical map as it pertains to the nonlinearity depends only on: (i) the sector bound of $\breve{T}$; (ii) each member of $\breve{T}$ being a proper static map; and (iii) the steepest slope among the members.  It is worth noting that if you repair the multivalued defect for one member of $\breve{T}$ using first\nobreakdash-event\nobreakdash-triggering with latching, then all members are repaired.  The control conditioning method that we will discuss in the Part II article smooths all members of the QS mapping from first\nobreakdash-event\nobreakdash-triggering with latching.

\subsection{Dynamical Mapping}
\label{sec:dyn_map}
The dynamical mapping also degrades due to interference.  The dynamical mapping is constructed to have the following properties:
\begin{enumerate}
\item Models the dominant dynamics that are degraded by interference. This degradation can result in instability and poor transient performance.
\item Input and output variables such that bounded input, bounded output (BIBO) stability in a sampled\nobreakdash-data space guarantees Lyapunov stability of the continuous time system.
\item Transient performance can be derived from the input and output variables of the mapping.
\item Separable as: (a) linear dynamics and 
(b) time-invariant or time\nobreakdash-varying non\nobreakdash-linearity.
\item Input and output variables of the mapping, linear dynamics, and non\nobreakdash-linearity correspond to physical quantities.
\item The switching period is much smaller than the $RC$ time constant.
\item The switching period is much smaller than the $L/R$ time constant.\!\footnote{A quantitative bound can be found in Appendix\,\ref{appendix:cmc_modeling}.}
\end{enumerate}

The stability of the continuous\nobreakdash-time (CT) dynamics of a power converter depend on the dynamical mapping.
For example, in examining the CT dynamics after a step change in the current command, the inductor current moves towards a new steady-state trajectory. This steady\nobreakdash-state trajectory is determined not only by the current command, but also the interference. The inductor current might remain stable and settle to a steady\nobreakdash-state trajectory, but can also become unstable. This unstable trajectory can manifest as an unbounded divergence, a limit cycle, or even chaos.  For a stable dynamical mapping, the effect of interference must decrease as the output trajectory approaches steady state, hence guaranteeing that the inductor current settles to a steady\nobreakdash-state CT trajectory.  The amount of time for the inductor current to reach steady state depends on how fast the effect of the interference decays in the dynamical mapping.

A current sensor output that is not a ramp, but is monotonic, can be conditioned to be a
{\em smooth}, but possibly {\em nonlinear static mapping}.  The interference can prolong the settling time or destabilize the dynamical mapping. Instabilities often manifest as subharmonics in the switching sequences of the current control loop; this cannot be represented by continuous\nobreakdash-time averaging models.
Even though subharmonics occur below the switching frequency, the features in the dynamics that cause it, arise within a switching period. Neither cycle\nobreakdash-by\nobreakdash-cycle nor sliding window averaging can be used to stabilize subharmonic phenomena. Subharmonic phenomena is a result of the dynamics of the peak current, which cannot be recovered from averaging.  In other words, we cannot recover the peak inductor current trajectory from the averaged trajectory, specifically because the subharmonic frequencies cannot be determined a\nobreakdash-priori.

We use a sampled-data model with a non\nobreakdash-uniform sampling of the peak current \cite{Cui2018a}. Because the instance of the peak inductor current is not in general periodic (e.g. during a transient), a sampled\nobreakdash-data model that represents this peak current at every cycle is necessarily non\nobreakdash-uniform in time. Hence, there is no continuous\nobreakdash-time equivalent to perform the stability and control performance analysis.

\begin{figure}
    \centering
    \includegraphics[width = 8cm]{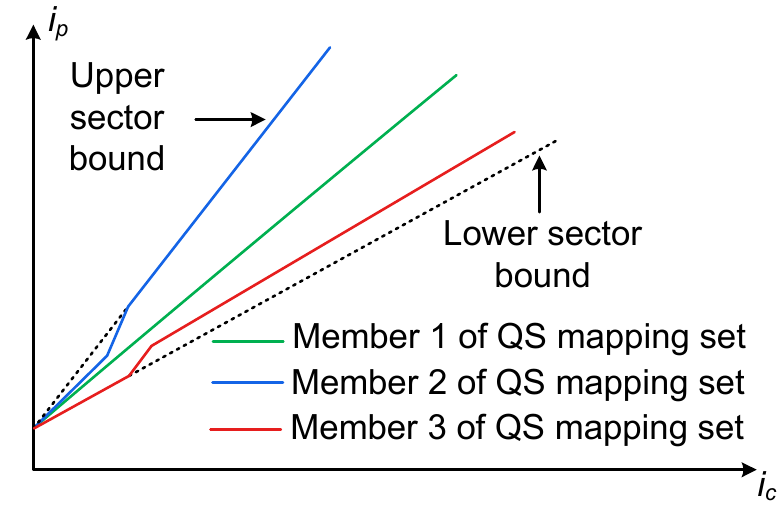}
    \caption{The quasi\nobreakdash-static (QS) current mapping is a family of functions from the current command $i_c$ to the actual peak inductor current $i_p$ and applies to a class of interference signal that does not repeat every switching period.}
    \label{fig:qsmapping}
\end{figure}
\begin{figure}
    \centering
    \includegraphics[width = 8cm]{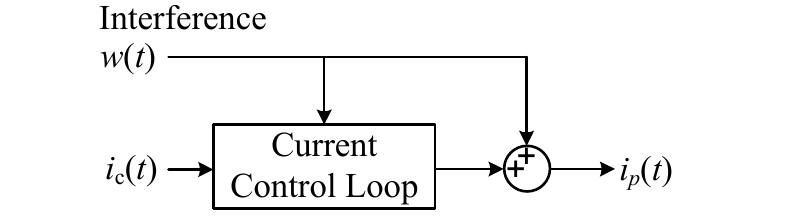}
    \caption{The non\nobreakdash-repeating interference $w(t)$ has two  effects: (i) an uncertainty in the current control loop and (ii) a feedthrough to the peak current output $i_p(t)$.}
    \label{fig:interference_feedforward}
\end{figure}


In this non\nobreakdash-uniform sampled\nobreakdash-data model, analysis must be performed directly from difference equations in the non\nobreakdash-uniform sampled\nobreakdash-data space. The current control loop consists of a current sensor $H_I$ and a current controller $f(\cdot)$, as shown in Fig.\,\ref{fig:dynmap}.
$f(\cdot)$ can be modeled in the sampled\nobreakdash-data space with the peak inductor current $i_p[n]$ and on time $t_{\text{on}}[n]$ as the state variables.
A useful way to represent the current controller is shown in Fig.\,\ref{fig:distrej}. Because a hardware comparator is used, the peak current is deadbeat to the current command.  For the circuit shown in Fig.\,\ref{fig:distrej}, with neither signal nor control conditioning, the transfer function \mbox{$G(z)=(1-z^{-1})$} from the perturbed peak current $\tilde{i}_p[n]$ to the perturbed on time $\tilde{t}_{\text{on}}[n]$ is deadbeat.


Interference is conventionally viewed as a sensor disturbance, often with implicit assumptions of a bounded amplitude, a lower bound on frequency, and possibly having memory.
The representation in Fig.\,\ref{fig:distrej} models the traditional sensor disturbance behavior.  This disturbance is homomorphic to an input to the current controller as shown in Fig.\,\ref{fig:distrej}.

Traditional remediation of interference includes filtering and blanking.
Filtering not only attenuates the interference $w[n]$, but also adds dynamics to $G(z)$, making the current controller slower.  When the dominant components of the disturbance are near the switching frequency, filtering falters because the filter dynamics have too similar of a time scale as the switching period making the filter dynamics dominate the current controller behavior. When the filter time constant extends over several switching periods, the notion of peak current mode control is no longer viable.  Blanking fails when the interference does not decay quickly enough within the switching period.

In contrast, the control conditioning approach treats interference as a model uncertainty, which in this paper is formulated as a Lure system whose stability is provable with the circle criterion \cite{Jury1964}, as shown Fig.\,\ref{fig:unctain}.  In particular, $I_c$ to $I_p$ is the open\nobreakdash-loop mapping of the dynamical mapping of the current controller in Fig.\,\ref{fig:distrej}, which corresponds to the static current mapping in Section\;\ref{sec:static_map}.
The mapping from $\tilde{i}_c[n]$ to $\tilde{t}_{\text{on}}[n]$ and $\tilde{i}_{c}[n]$ is the {\em dynamical mapping} of the deviation of the response from the open loop. $\psi(\cdot)$ embeds the previously discussed static mapping $\mathcal{T}$ from the current command $i_c$ to $i_p$,
\begin{subequations}
    \begin{align} 
        i_c  = I_c & + \tilde{i}_c,   \\
        i_p  = I_p & + \tilde{i}_p,   \\
        \mathcal{T}:i_c & \rightarrow i_p  \label{eqn:defT},\\
        \tilde{\mathcal{T}}: \tilde{i}_c & \rightarrow \tilde{i}_p \label{eqn:deftildeT},\\
        \tilde{\mathcal{T}}^{-1}(x) & = x + \psi(G_0 x) \label{eqn:T2psi},\\
        \frac{\text{d}\,(\tilde{\mathcal{T}}^{-1})}{\text{d}\,x}  & = 1 + G_0\psi^{'} \label{eqn:devTinvx},\:
        \text{where}\; \psi^{'} = \frac{\text{d}\,\psi(x)}{\text{d}\,x},\\
       \frac{\text{d}\,\mathcal{T}}{\text{d}\,x} &= \frac{\text{d}\,\tilde{\mathcal{T}}}{\text{d}\,x} = \frac{1}{1 + G_0\psi^{'}}\label{eqn:devTx},\\
         \psi^{'} &  = \frac{1}{G_0} \left( \frac{1}{\frac{\text{d}\,\mathcal{T}}{\text{d}\,x}} - 1  \right)  \label{eqn:devvarphi},
    \end{align}
\end{subequations}
where $x$ in this case (i.e., peak current-mode control) is the deviation of peak inductor current from the open loop and $G_0$ is the open-loop mapping from $\tilde{i}_p$ to $\tilde{t}_{\text{on}}$. The static current mapping $\mathcal{T}$ from the current command to the peak current is necessarily monotonic because it is proper, which makes it invertible in (\ref{eqn:T2psi}). $\psi$ is in general nonlinear; however, a necessary condition to prove stability of a Lure system requires $\psi$ to be locally Lipschitz and sector\nobreakdash-bounded. 
If $\mathcal{T}$ is monotonic (which can be repaired by control conditioning) and locally Lipschitz (which can be enforced by ensuring an upper bound on the bandwidth and amplitude of the interference either by  control conditioning, mild signal conditioning/filtering, or from the underlying physical mechanisms of the interference), then $\psi$ is locally Lipschitz.  The bound on interference amplitude also ensures that $\psi$ is sector-bounded.  Without control conditioning, which modifies $G(z)$ and repairs $\psi$, the feedback path through $\psi$ can cause instability.

For a quasi\nobreakdash-static mapping, the Lure formulation still applies because $\psi$ can also be time\nobreakdash-varying \cite{Brockett1966}. QS mapping $\breve{T}$ is the set of mappings $\breve{T}[n]$ from $i_c$ to $i_p$
\begin{align}
    T = \big\{ \breve{T}[n] \in  \breve{T} \mid n \in \mathbb{N}, \breve{T}[n]: i_c \in \mathbb{R} \mapsto i_p \in \mathbb{R} \big\}.
\end{align}
There exists a mapping $\mathcal{M}$ which maps $\breve{T}[n]$ to a time-varying nonlinearity $\psi[n]$ in the dynamical mapping
\begin{align}
   \mathcal{M}: \breve{T}[n] \mapsto \psi[n].
\end{align}

When the dynamical mapping is a Lure system \cite{khalil2002nonlinear}, powerful control-theoretic tools for nonlinear systems can guarantee absolute large\nobreakdash-signal stability, performance limits and bounds, and robustness.
This insight allows us to analyze how interference destabilizes or degrades the transient performance of the dynamical mapping.
This model for the dynamical mapping enables the precise design of control conditioning where we co-design $i_c$, $G$, and $\psi$ simultaneously. In contrast, signal conditioning only modifies $i_c$. The differences between the traditional and proposed views on treating interference are summarized in Fig.\,\ref{fig:distrejvsrobustcontrol}.
\begin{figure}
    \centering
    \includegraphics[width = 6cm]{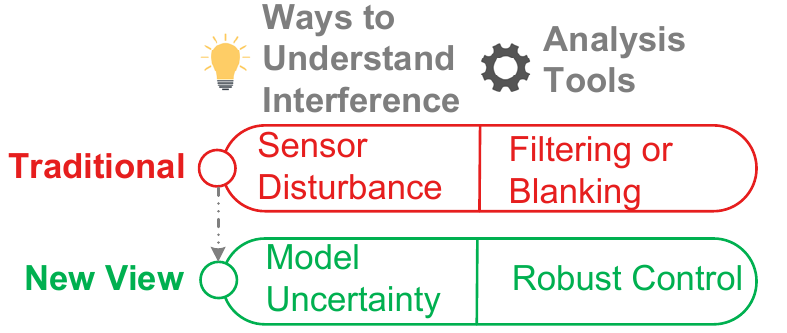}
    \caption{Compared to the traditional view of understanding interference as a model uncertainty in a Lure system so that the robust control tools can be applied.}
    \label{fig:distrejvsrobustcontrol}
\end{figure}

\begin{figure}[ht]
\centering
\subfigure[Traditionally, the interference is viewed as a sensor disturbance. 
The disturbance causes the deviation from the peak current command.
$w$ represents this deviation at every instance $n$.
The stable plant $G(z)$ represents the current controller using constant off\nobreakdash-time.
This model cannot explain the destabilization effect of interference to the current controller and cannot accurately predict the transient performance of the current controller in the presence of interference.
]{
    \centering
    \includegraphics[width = 8cm]{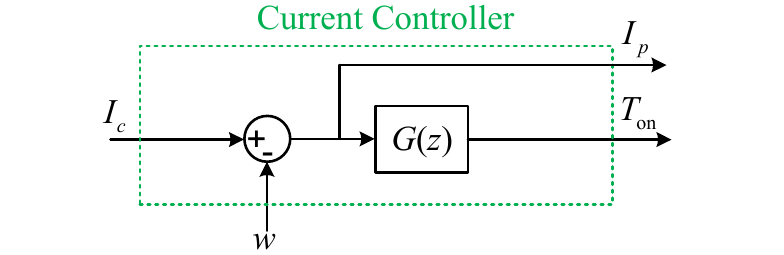}
    \label{fig:distrej}
}
\subfigure[We view the interference as a model uncertainty. In this Lure system, both stable plant $G(z)$ and interference $\psi$ determine the transient performance of the current control loop using constant off\nobreakdash-time. 
This model can explain the destabilization effect of interference to the current control loop.
This model can accurately predict the transient performance of the current control loop in the presence of interference.]{
    \centering
    \includegraphics[width = 8cm]{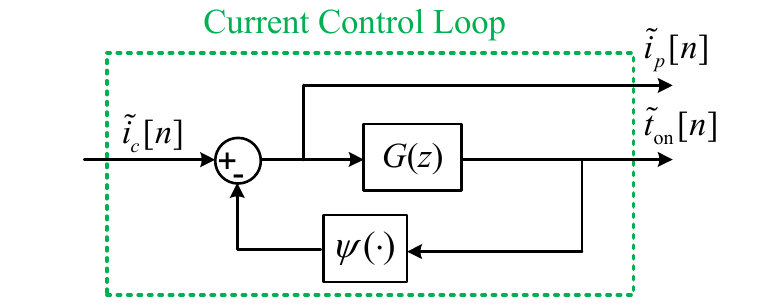}
    \label{fig:unctain}
}
\caption{Comparison of the traditional view and our view of treating interference. Modeling power converters as Lure systems allows us to guarantee large-signal stability.}
\end{figure}

\subsection{Current Mappings Predict Control Performance} \label{current_map_ctrl_perform}

We take advantage of three physical characteristics of interference to obtain a rigorous mathematical definition.
\begin{definition} \label{def:inteference} 
$w(t)$ is an {\em interference function} in continuous time, which can be considered as a deviation to the ideal signal satisfying the following properties: \begin{enumerate*} [label=(\roman*)] \item additive; \item bandwidth limited by $\omega_{ub}$ with no dc component; and \item Fourier transform  $W(\omega)$ is absolutely integrable.
\end{enumerate*}
\end{definition}
For (i), we discussed in Section \ref{sec:ctrlcond_intf} that interference is a deviation from the ideal signal; deviations are additive to the ideal signal.
For (ii), sensors in power electronics typically use a low\nobreakdash-pass filter; this automatically bandwidth-limits the signal from the current sensor.  The dc component of interference does not affect the dynamics of the current loop.  Because of this, the dc component can be moved outside of the current loop.  Outside of the current loop, the dc component appears as an offset to the current command. In practice, a voltage output converter using current-mode control corrects this current command offset using integral control in the voltage-feedback loop. Because the dc component is outside the current loop, only the ac component of interference needs to be considered in the current control loop analysis.
For (iii), absolute integrability of physical states is consistent with a broad class of interference from physical systems whose energy decays.
The absolute integrability condition guarantees the existence of the inverse Fourier transform of $W(\omega)$ as well as the bound on amplitude $A_{ub}$. Together, absolute integrability and limited bandwidth guarantees the existence of the upper bound on the Lipschitz constant ${\Lambda}_{ub}$.
\begin{align} \label{eqn:aub}
      |w(t)| = & \Bigg|\int_{-\infty}^{+\infty} W(\omega)e^{j\omega t}\,\text{d}\omega \Bigg| \nonumber \\
      \le &  \int_{-\infty}^{+\infty} \big| W(\omega) \big| \,\text{d}\omega \triangleq A_{ub}.
\end{align}
\begin{align}  \label{eqn:lub}
     |w'(t)| = & \Bigg|\int_{-\infty}^{+\infty} j\omega W(\omega) e^{j\omega t}\,\text{d}\omega \Bigg| \nonumber \\
     \le & \int_{-\infty}^{+\infty} \big | \omega W(\omega) \big| \,\text{d}\omega \triangleq {\Lambda}_{ub}.
\end{align}
Transient performance depends on the dynamical mapping.
We examine two types of power converters that  use current-mode control: variable switching frequency and fixed switching frequency.  
In these analyses, the inductor current is composed of waveforms that are rising and falling ramps. The slope of the rising ramp is $m_1$ and the slope of the falling ramp is $m_2$. 
A class of variable frequency converters uses either constant off\nobreakdash- or constant on\nobreakdash-time.

\subsubsection{Large\nobreakdash-Signal Stability of the \textbf{Constant Off\nobreakdash-Time} Current Control Loop}
We start with a constant off-time current control loop as shown in Fig.\,\ref{fig:eboostschmatic} and Fig.\,\ref{fig:eboostwaveform} to exemplify the control-conditioning analysis.

Peak inductor current is what is controlled for 
constant off\nobreakdash-time converters.
The current control loop with interference can be modeled as
\begin{subequations}
\label{eqn:cmcot}
    \begin{align} 
    \label{eqn:cmcotip} i_p[n] &= i_p[n-1] - m_2T_{\text{off}} +m_1t_{\text{on}}[n], \\
    \label{eqn:cmcotic} i_c [n] &= i_p[n] +w (t_{\text{on}}[n]).
    \end{align}
\end{subequations}

The current command from the outer loop in every cycle is denoted by $i_c[n]$. The peak inductor current in every cycle is denoted by $i_p[n]$.  $T_{\text{on}}$ is the equilibrium when $\psi = 0$, and $I_p$ is the equilibrium peak inductor current corresponding to $T_{\text{on}}$. $I_p$ and $T_{\text{on}}$ can be obtained by letting \mbox{$i_p[n+1] = i_p[n]$}, \mbox{$t_{\text{on}}[n+1] = t_{\text{on}}[n]$}. The relationship of the $I_p$ and $I_c$ is described by the static current mapping, which we denote by $\mathcal{T}$. 

Given an equilibrium at $i_p[n] = I_p$, $t_{\text{on}}[n]= T_{\text{on}}$, $i_c[n] = I_c$,  we can translate the equilibrium to the origin by defining \mbox{$\tilde i_p[n] = i_{p}[n] - I_{p}$}, \mbox{$ \tilde t_{\text{on}}[n] = t_{\text{on}}[n] -  T_{\text{on}}$}, \mbox{$ \tilde i_c[n] = i_c[n] - I_c$}. The system represented by (\ref{eqn:cmcot}) can now be expressed as
\begin{subequations}
\label{eqn:cmcotshift}
    \begin{align}
     \label{eqn:cmcotipshift} \tilde i_p[n] & = \tilde i_p[n-1] +  m_1 \tilde t_{\text{on}}[n], \\
     \label{eqn:cmcoticshift} \tilde i_c [n] & = \tilde i_p[n] +w(t_{\text{on}}[n])- w(T_{\text{on}}).
    \end{align}
\end{subequations}


System (\ref{eqn:cmcotipshift}) and  (\ref{eqn:cmcoticshift}) can be transformed into a Lure system as shown in Fig.\;\ref{fig:luresystem} when $w(t_{\text{on}}[n]) - w(T_{\text{on}})$ is sector\nobreakdash-bounded.
This interference\nobreakdash-related term $w(t_{\text{on}}[n]) - w(T_{\text{on}})$ is treated as a model uncertainty and expressed as a nonlinear function $\psi(\cdot)$.  
$G(z)$ is a linear transfer function in the sampled\nobreakdash-data space and $\psi(\cdot)$ is a sector-bounded nonlinear function.
Note that although the peak inductor current sequence does not have a uniform correspondence to the physical time domain, the $z$\nobreakdash-transform can still be applied \cite{Cui2018a}.
By applying the circle criterion \cite{Brockett1966}, we can prove the stability condition for the current control loop described by (\ref{eqn:cmcotipshift}) and (\ref{eqn:cmcoticshift}) for large perturbations.

Theorem 1 shows that the stability of the current control loop using constant off\nobreakdash-time is only dependent on the Lipschitz constant of $\psi(\cdot)$ and the rising slope $m_1$ of the inductor current.
We highlight several properties of this stability theorem: (1) It is a global stability criterion. The result holds for arbitrarily large interference; 
(2) This theorem is especially useful when the interference does not identically repeat every cycle.
\begin{theorem} \label{theorem:gloasystab}
The current control loop represented by the Lure system in Fig.\,\ref{fig:luresystem} is \emph{globally asymptotically stable} if \mbox{${\Lambda}_{ub} < m_1/2$}.
\end{theorem}
\begin{figure*}[ht]
\centering
\subfigure[The Lure system representation
\newline for large-signal analysis.
\newline The interference is embedded in $\psi(x)$.]{
\includegraphics[width = 0.4\textwidth ]{unctain.pdf} 
   \label{fig:luresystem}
}
\subfigure[Large\nobreakdash-signal Nyquist plot for $G(z)$.\newline The  circle criterion is used for determining stability.]{
    \includegraphics[width = 0.3 \textwidth]{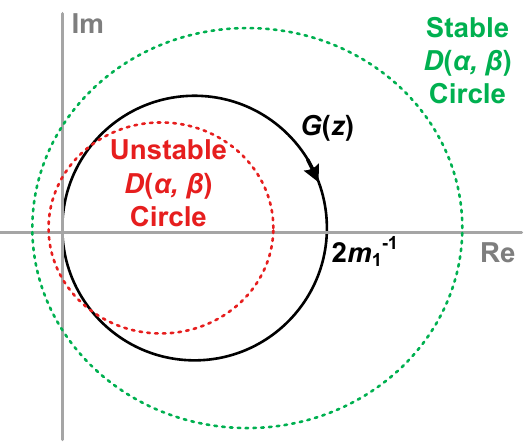} 
    \label{fig:circlecriterionplot}
    }
\subfigure[Small-signal root locus plot for $G(z)$.
\newline
{\color{green}---} is for the negative feedback;
\newline
{\color{red}---} is for the positive feedback.]{
    \includegraphics[width = 0.2 \textwidth]{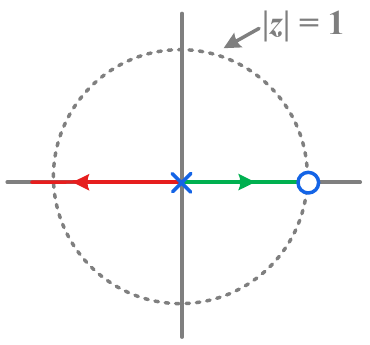} 
    \label{fig:rlpurecotcm}
}
\caption{Modeling of the constant off\nobreakdash-time current control loop with interference. Lure systems can be used for analyzing both large\nobreakdash-signal and small\nobreakdash-signal stability in variable\nobreakdash-frequency current\nobreakdash-mode power converters.}
\end{figure*}
\begin{proof}
The sufficient condition can be proven by the Tsypkin criterion (i.e. discrete-time circle criterion) \cite{Larsen2001}.
System\,(\ref{eqn:cmcotshift}) can be reformulated as a Lure system as shown in Fig. \ref{fig:luresystem} with 
\begin{align}
    \label{eqn:gzofcotcm} G(z) &= \frac{1-z^{-1}}{m_1},\\
    \psi(x) &= w(x + T_{\text{on}}) - w(T_{\text{on}}).
\end{align}

From the circle criterion \cite{Okuyama2014}, assume $\psi$ is in sector $[\alpha, \beta]$, the Lure system is globally asymptotically stable if the Nyquist plot of $G(z)$, as shown in Fig.\;\ref{fig:circlecriterionplot}, lies in the interior of disk $D(\alpha,\beta)$. 
$D(\alpha, \beta)$ is defined as the closed disk in the complex plane whose diameter is the line segment connecting points $-1/\alpha$ and $-1/\beta$.

We observe that for all disk $D(\alpha,\beta)$ with $\alpha \ge -m_1/2$ and $\beta < \infty$, $G(z)$ lies inside the $D(\alpha,\beta)$ from Fig.\;\ref{fig:circlecriterionplot}.
This directly results in the sufficient stability criterion that ${\Lambda}_{ub} < m_1/2$. ${\Lambda}_{ub}$ is the least upper bound of the Lipschitz constant.

Remark that this stability criterion does not require the $\psi(x)$ to be same in every switching cycle.
This means the stability condition still holds even if the interference is not repetitive in every switching cycle.

\end{proof}



An intuition for the stability of (\ref{eqn:cmcotipshift}) and (\ref{eqn:cmcoticshift}) can be discerned by linearization
\begin{subequations}
\label{eqn:cmcotshiftlin}
    \begin{align}
        \label{eqn:cmcotipshiftlin} \tilde i_p[n] &= \tilde i_p[n-1] +  m_1 \tilde t_{\text{on}}[n], \\
        \label{eqn:cmcoticshiftlin} \tilde i_c [n]  & = \tilde i_p[n] + \frac{\text{d}w}{\text{d}t}\Bigg \rvert_{ t = T_{\text{on}}} \tilde t_{\text{on}} [n].
    \end{align}
\end{subequations}

Without interference, 
the feedback path in Fig.\;\ref{fig:luresystem} disappears and the current control loop can be considered an open\nobreakdash-loop system. The current control loop is stable because the pole of $G(z)$ is located at $z=0$, which is inside the unit circle.
This is consistent with the fact that the constant off\nobreakdash-time control guarantees a stable inductor current for all duty cycles.
With interference, the current control loop needs to be treated as a closed-loop system. 
The gain in the feedback path is the worst\nobreakdash-case linearized slope of $\psi$. 
If this feedback gain magnitude is positive, the current control loop is in negative feedback and
the root locus starts from the open\nobreakdash-loop pole at $z=0$ and ends at the open\nobreakdash-loop zero at $z=1$, as shown in the solid green line in Fig.\;\ref{fig:rlpurecotcm}.
Note that if the current control loop is in negative feedback, no matter how large the feedback gain magnitude, the current control loop is always stable. 

If the feedback gain is negative, the current control loop is in positive feedback, the root locus starts from the open-loop pole at $z=0$ and ends at $z= -\infty$ as shown in the dashed blue line in Fig.\;\ref{fig:rlpurecotcm}.
When the magnitude of the negative gain is too high, the closed\nobreakdash-loop pole moves outside of the unit circle, which makes the current control loop unstable.
The boundary between stability and instability for the current control loop occurs when the feedback gain is $-m_1/2$.

\subsubsection{Transient Performance of the \textbf{Constant Off\nobreakdash-Time} Current Control Loop}
Stability is the minimum needed for any power converter.
Beyond this, settling time and overshoot in continuous physical time are two important transient criteria.

Typically, these performance criteria are evaluated in the linearized system.
We perform the analysis on these performance criteria in the sampled\nobreakdash-data space and examine the correspondence to continuous physical time.
Theorem 2 in \cite{xiaofanacc2019} rigorously proves in minimizing the settling cycles in the sampled\nobreakdash-data space, the physical settling time is minimized as well.  
Also, the overshoot in continuous physical time is bounded from above by the overshoot in the sampled-data space, from Theorem 3 in \cite{xiaofanacc2019}.

Because the current control loop is a first\nobreakdash-order system in the sampled\nobreakdash-data model, we can analytically derive the difference equation that describes the closed\nobreakdash-loop system. 
The linearized difference equation of the system (\ref{eqn:cmcotipshiftlin}) and (\ref{eqn:cmcoticshiftlin}) is
\begin{align} \label{eqn:cmcotlin_tf_td}
\tilde i_p[n] &=  \beta \tilde i_c [n] + a \tilde i_p[n-1],
\end{align}
where
\begin{align}
    \beta = \frac{1}{1 + \sigma}, \quad a = \frac{\sigma}{1+\sigma}, \quad \sigma = \frac{1}{m_1}\frac{\text{d}w}{\text{d}t} \Bigg \rvert_{t = T_{\text{on}}}. \nonumber
\end{align}
The parameter $\sigma$ can be interpreted as the normalized small\nobreakdash-signal gain from the interference. 
We apply the $z$-transform to (\ref{eqn:cmcotlin_tf_td})
\begin{align}  \label{eqn:cmcotlin_tf_zd}
    C(z) = \frac{\beta}{1- a z^{-1}}.
\end{align}

The transfer function in the $z$\nobreakdash-domain allows the analysis of the settling cycles and overshoot. These two performance criteria depend on the pole locations of the transfer function.
If no interference appears ($\sigma = 0$), $\beta$ equals 1 and the pole $a$ is at the origin. 
This indicates that the output time trajectory exactly follows the input time trajectory in every cycle.

With interference, the pole $a$ cannot be located exactly, but rather within a range
\begin{align} \label{label:arange1}
a_{\text{min}} \triangleq 1 - \frac{m_1}{(m_1 - {\Lambda}_{ub})} \le a \le a_{\text{max}} \triangleq 1 - \frac{m_1}{(m_1 + {\Lambda}_{ub})}. \vspace{+2pt}
\end{align}
It is worth noting that the location of $a$ is uncertain because the phase of each frequency component of the interference is uncertain.
We derive the worst\nobreakdash-case settling $N_w$ from the location of the pole as\footnote {Note that the settling might not be an integer. To obtain the settling cycles in the sampled-data space, we can round this number to the next higher positive integer. 
In this paper, we directly use the settling $N_w$ as the metric for the \emph{settling cycles}.}
\begin{align} \label{eqn:settlecycle1}
    N_w \triangleq \text{max} \bigg\{\bigg|\frac{4}{\text{ln}(|a_{\text{min}}|)}\bigg|,\bigg|\frac{4}{\text{ln}(|a_{\text{max}}|)}\bigg|\bigg\}.
\end{align}

The step response of system (\ref{eqn:cmcotlin_tf_zd}) has an overshoot if \mbox{$a_{\text{min}}<0$}. The \emph{worst-case overshoot} $O_w$ is
\begin{align} \label{eqn:overshoot1}
    O_w \triangleq \text{max}\{-a_{\text{min}},0\}, 
\end{align}
where $O_w$ is expressed in percentage form in this paper.
The stability and transient performance results of the constant off-time current control loop are summarized in Appendix\,\ref{ref:sec_design_eqn}, Table\,\ref{table:cont_off_cm_theory}.
\subsubsection{Large\nobreakdash-Signal Stability and Transient Performance of the \textbf{Constant On\nobreakdash-Time} Current Control Loop}
The analysis in the previous section can be applied to constant on\nobreakdash-time converters where valley current is controlled.

We can substitute $m_2$ for $m_1$ in Theorem\,\ref{theorem:gloasystab} to derive the stability condition and in (\ref{label:arange1}), (\ref{eqn:settlecycle1}), and (\ref{eqn:overshoot1}) for the settling and overshoot.
For example, the stability criterion of the current control loop using constant off\nobreakdash-time is given by Corollary 1.
\begin{corollary}
    The constant on-time current control loop is \emph{globally asymptotically stable} if ${\Lambda}_{ub} < m_2/2 $.
\end{corollary}
The stability and transient performance results of the constant on\nobreakdash-time current control loop are summarized in Appendix\,\ref{ref:sec_design_eqn}, Table\,\ref{table:cofft_on_cm_theory}.
This similar analysis can be
applied to any buck, boost, buck\nobreakdash-boost, and other derived converters that use constant on\nobreakdash- or off\nobreakdash-time control.

\subsubsection{Large-Signal Stability of the \textbf{Fixed-Frequency Peak} Current Control Loop}
The analysis framework can also be applied to the \mbox{dc\nobreakdash-dc} converters using fixed\nobreakdash-frequency control. We take the fixed\nobreakdash-frequency peak current\nobreakdash-mode control in Fig.\;\ref{fig:fixedfreqpccboost}, \ref{fig:fpccwaveform} as an example.

The current command from the outer voltage loop is represented by $i_c[n]$. The duty cycle is represented by $d\,[n]$ and \mbox{$d^{'}[n] = 1-d[n]$}. $T$ represents the switching period. 
In general, if the equilibrium is at \mbox{$d\,[n]= D$}, we can translate the equilibrium to the origin.
The generalized state\nobreakdash-space model of the current control loop is
\begin{align} 
\label{eqn:ffpccip_ss} \tilde i_v[n] = \,&\tilde i_v[n-1] + (m_1 + m_2)\,\tilde d\,[n-1]\,T, \\
\label{eqn:ffpccic_cf} 0 =\,& \tilde i_v[n] +  w(\tilde d\,[n]\,T + D\,T) - w(D\,T) \nonumber \\
& + m_1 \tilde{d}[n] T- \tilde i_c [n],\\
\label{eqn:ffpccic_output} \tilde{i}_p[n] = \,& \tilde{i}_v[n] + m_1\tilde{d}[n]T ,
\end{align}
where $\tilde{i}_v[n]$ is the state variable, $\tilde{d}[n]$ is the auxiliary state variable, $\tilde{i}_c[n]$ is the input, and $\tilde{i}_p[n]$ is the output. (\ref{eqn:ffpccip_ss}) is the conventional state\nobreakdash-space equation, (\ref{eqn:ffpccic_cf}) is the constraint function, and (\ref{eqn:ffpccic_output}) is the output equation.

We reformulate the current control loop by substituting (\ref{eqn:ffpccic_output}) into (\ref{eqn:ffpccip_ss}) and (\ref{eqn:ffpccic_cf}) as
\begin{subequations}
\label{eqn:ffpcc} 
    \begin{align} 
    \label{eqn:ffpccip} \tilde i_p[n] &= \tilde i_p[n-1] -m_2 \, \tilde d \,[n-1] \, T + m_1\,\tilde d\,[n]\,T, \\
    \label{eqn:ffpccic} 0 &= \tilde i_p[n] +  w(\tilde d\,[n]\,T + D\,T) - w(D\,T) - \tilde i_c [n].
    \end{align}
\end{subequations}

We can view system (\ref{eqn:ffpcc}) as a Lure system as shown in Fig.\;\ref{fig:luresystefpcc}
with
\begin{align} \label{eqn:fpccplanttf}
    G(z) &= \frac{1-z^{-1}}{m_1 + m_2 z^{-1}},\\
    \psi(x) &= w(x + DT) - w(DT).
\end{align}

Unlike with variable frequency current loops where the open-loop pole from $G(z)$ is always at the origin, in constant frequency current  loops, the pole location depends on both $m_1$ and $m_2$ and would be unstable if located outside the unit circle.

By applying a proof similar to that in Theorem \ref{theorem:gloasystab}, we derive the stability condition for the system (\ref{eqn:ffpccip}) and (\ref{eqn:ffpccic}) under any large-signal disturbance.
Theorem \ref{theorem:gloasystabfpcc} shows that the stability of the current control loop is dependent on the upper bound of the Lipschitz constant of the interference, the rising slope of the inductor current, and the falling slope of the inductor current.
We highlight several properties of this stability theorem: (1) it is a global stability criterion. The result holds for arbitrarily large interference; 
and (2) this theorem is especially useful when the interference does not identically repeat every cycle.
\begin{theorem} \label{theorem:gloasystabfpcc}
The current control loop represented by the Lure system in Fig.\,\ref{fig:luresystefpcc} is \emph{globally asymptotically stable} if \mbox{$\Lambda_{ub} < (m_1 - m_2)/2$}.
\end{theorem} 
We observe that if $m_2 > m_1$, the system loses stability even with zero interference (\mbox{${\Lambda}_{ub} = 0$}). This observation matches the well-recognized result that fixed-frequency peak current-mode control is unstable for $D>0.5$.
An intuition for the stability of the current control loop can be discerned by linearization.
The linearization of system (\ref{eqn:ffpccip}) and (\ref{eqn:ffpccic}) results in
\begin{align} 
\label{eqn:ffpccipshift} \tilde i_v[n] = \,&\tilde i_v[n-1] + (m_1 + m_2)\,\tilde d\,[n-1]\,T, \\
\label{eqn:ffpccicshift}  0 =\,& \tilde i_v[n] +  w^{'}(D T) \,\tilde d[n] T \nonumber \\
& + m_1 \tilde{d}[n] T- \tilde i_c [n],\\
\label{eqn:ffpcc_output} \tilde{i}_p[n] = \,& \tilde{i}_v[n] + m_1\tilde{d}[n]T.
\end{align}


\begin{figure*}[t]
\centering
\subfigure[The Lure system representation \newline for large-signal analysis. \newline The interference is embedded in $\psi(x)$.]{
\includegraphics[width=0.4\textwidth ]{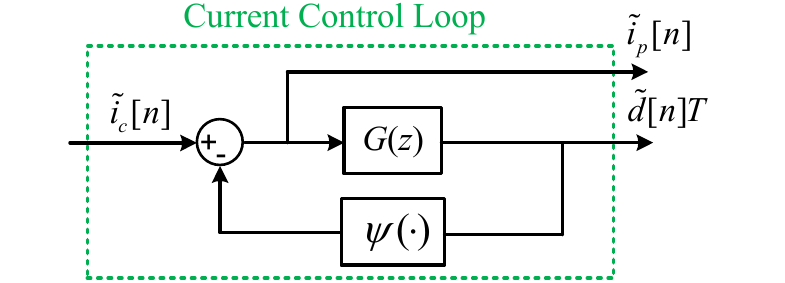} 
     \label{fig:luresystefpcc}
}
\subfigure[Large-signal Nyquist plot for $G(z)$. 
\newline The circle criterion is used for determining \newline stability.]{
    \includegraphics[width=0.3\textwidth ]{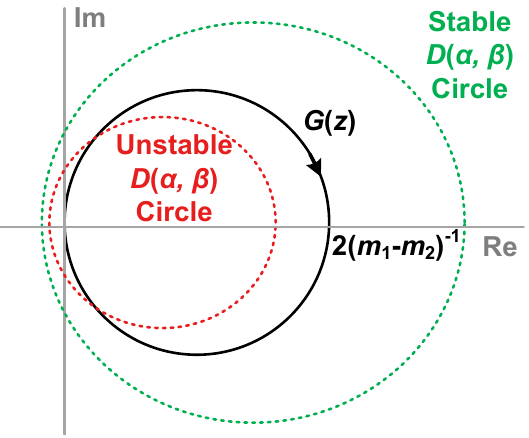} 
    \label{fig:circlecriterionplotfpcc}
    }
\subfigure[Small-signal root locus plot for $G(z)$. \newline {\color{green}---} is for the negative feedback; \newline
{\color{red}---} is for the positive feedback.]{
    \includegraphics[width=0.2\textwidth]{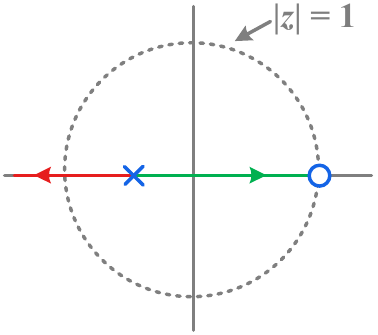} 
    \label{fig:rl_fpcc}
}
\caption{Modeling of the fixed-frequency current control loop with interference. Lure systems can be analyzed for both large-signal and small-signal stability in fixed-frequency current-mode power converters.}
\end{figure*}
The gain in the feedback path is $\psi^\prime(x)$, whose effect can be observed in the small\nobreakdash-signal root locus plot in Fig.\;\ref{fig:rl_fpcc}.
If \mbox{$m_2 > m_1$}, the open-loop pole \mbox{$z = m_2/m_1$} is outside of the unit circle and the current control loop is unstable even when the interference function is zero (i.e., \mbox{$\psi = 0$}).
If the current control loop is in negative feedback (\mbox{$\psi^\prime(x)>0$}), as gain magnitude increases, the root locus starts from the open\nobreakdash-loop pole at \mbox{$z = -m_2/m_1$} and and approaches the open\nobreakdash-loop zero at \mbox{$z=1$}, as shown in the solid green line in Fig.\,\ref{fig:rl_fpcc}.
For any negative feedback gain, this system is always stable.
If the current control loop is in positive feedback (\mbox{$\psi^\prime(x)<0$}), as gain magnitude increases, the root locus starts from the open\nobreakdash-loop pole \mbox{$z = -m_2/m_1$} and approaches \mbox{$z=-\infty$}, as shown in the dashed blue line in Fig.\,\ref{fig:rl_fpcc}.
For \mbox{$\psi^\prime(x) < (m_1 - m_2)/2$}, the current control loop is unstable.

\subsubsection{Transient Performance of the \textbf{Fixed-Frequency Peak} Current Control Loop}
For a stable converter, we focus on two important transient performance criteria --- settling and overshoot in the sampled-data space.
We solve the auxiliary variable $\tilde{d}[n]$ from the linearized constraint function (\ref{eqn:ffpccicshift}) and then substitute it into (\ref{eqn:ffpccipshift}) and (\ref{eqn:ffpcc_output}). 
The current control loop can then be described as an ordinary difference equation
\begin{align}
\tilde i_p[n] &=  \beta (\tilde i_c [n] -b \tilde i_c [n-1]) + a \tilde i_p[n-1],
\end{align}
where
\begin{align}
\beta = \frac{1}{1 + \sigma}, \quad a = -\frac{M - \sigma}{1 + \sigma}, \quad \sigma = \frac{1}{m_1}\frac{\text{d}\,w}{\text{d}\,t} \Bigg\rvert_{t = T_{\text{on}}},
\end{align}
\begin{align}
\quad b = -M, \quad M = \frac{m_2}{m_1} < 1. 
\end{align}
We apply the $z$-transform in this sampled-data space
\begin{align} \label{eqn:tf_fpcc}
    C(z) = \beta \, \frac{1 - b z^{-1}}{1 - a z^{-1}}.
\end{align}


We observe that $\sigma$, which depends on interference, affects the location of the poles and zeros of $C(z)$.
Larger $|a|$ means longer settling for the current control loop.

The closed-loop pole $a$ lies within the range given by
\begin{align} \label{eqn:arange_ff}
    a_{\text{min}} \triangleq \frac{-{\Lambda}_{ub} - m_2}{- {\Lambda}_{ub} + m_1} \le a \le a_{\text{max}} \triangleq \frac{{\Lambda}_{ub}- m_2}{{\Lambda}_{ub} + m_1}.
\end{align}
The worst-case settling is given by (\ref{eqn:settlecycle1}). Because (\ref{eqn:tf_fpcc}) is a single\nobreakdash-pole\nobreakdash-single\nobreakdash-zero system, The worst\nobreakdash-case overshoot is given by
\begin{align} \label{eqn:os1p1z}
    O_w \triangleq \text{max}\bigg\{\frac{b-a_{\text{min}}}{1-b}\,,\,0\bigg\}.
\end{align}
The stability and transient performance results of the fixed\nobreakdash-frequency peak current control loop are summarized in Appendix\,\ref{ref:sec_design_eqn}, Table\,\ref{table:ff_peak_cm_theory}.
\subsubsection{Large\nobreakdash-Signal Stability and Transient Performance of \textbf{Fixed\nobreakdash-Frequency Valley} Current Control Loop}
These stability and transient performance analyses can also be applied to fixed\nobreakdash-frequency valley current-mode control.

We can substitute $m_2$ for $m_1$ in Theorem\,\ref{theorem:gloasystabfpcc} to derive the stability condition. Settling and overshoot can be derived from (\ref{eqn:settlecycle1}), (\ref{eqn:arange_ff}), and (\ref{eqn:os1p1z}).
The stability and transient performance results of fixed\nobreakdash-frequency valley current control loop are summarized in Appendix\,\ref{ref:sec_design_eqn}, Table\,\ref{table:ff_valley_cm_theory}.

This similar analysis can be applied to any buck, boost, buck\nobreakdash-boost, and derived converters that use fixed-frequency peak or valley current control.
\section{Conclusion} \label{sec:conclusion_cmc_part1}
In this paper, we presented a rigorous framework for analyzing the effects of interference in current sensors on the control model for high-frequency extremum current-mode power converters and for guaranteeing large-signal stability.  In this framework, the current control loop models as a Lure system, where methods in nonlinear control can be applied.  This framework leads to a control conditioning approach, which is particularly important when the interference and switching frequencies are near each other.  Three control conditioning methods are presented in Part II of this paper series to stabilize the current control loop even when the interference is severe.




\section*{Acknowledgements}
Special thanks to Professor Peter Seiler for the contribution in large-signal stability analysis. This work is based upon work supported by the U.S. Department of Energy SunShot Initiative, under Award Number(s) DE-EE-0007549. 
\newpage
\begin{appendices}

\section{Large-Signal Stability of Current-Mode \\ DC-DC Converters} \label{appendix:cmc_modeling}
The large-signal modeling approach we will introduce applies to all types of current-mode dc-dc converters. In this section, we take the current-mode buck converter using constant on\nobreakdash-time and current-mode boost converter using constant off\nobreakdash-time as two typical examples \cite{Avestruz2022}.
\subsection{Current-Mode Buck Converter Using Constant On-Time}
Consider a class $\Sigma$ buck converter defined in \cite{Cui2018a}, the off\nobreakdash-time at steady state is denoted by $T_{\text{off}}$.
The time varying off\nobreakdash-time is bounded from above by
\begin{align} \label{eqn:toff_bd_buck}
    T^{\text{min}}_{\text{off}} \le t_{\text{off}}[n] \le T^{\text{max}}_{\text{off}},
\end{align}
The large-signal stability can be described by the following three propositions:

\begin{figure}
    \centering
    \includegraphics[scale = 1]{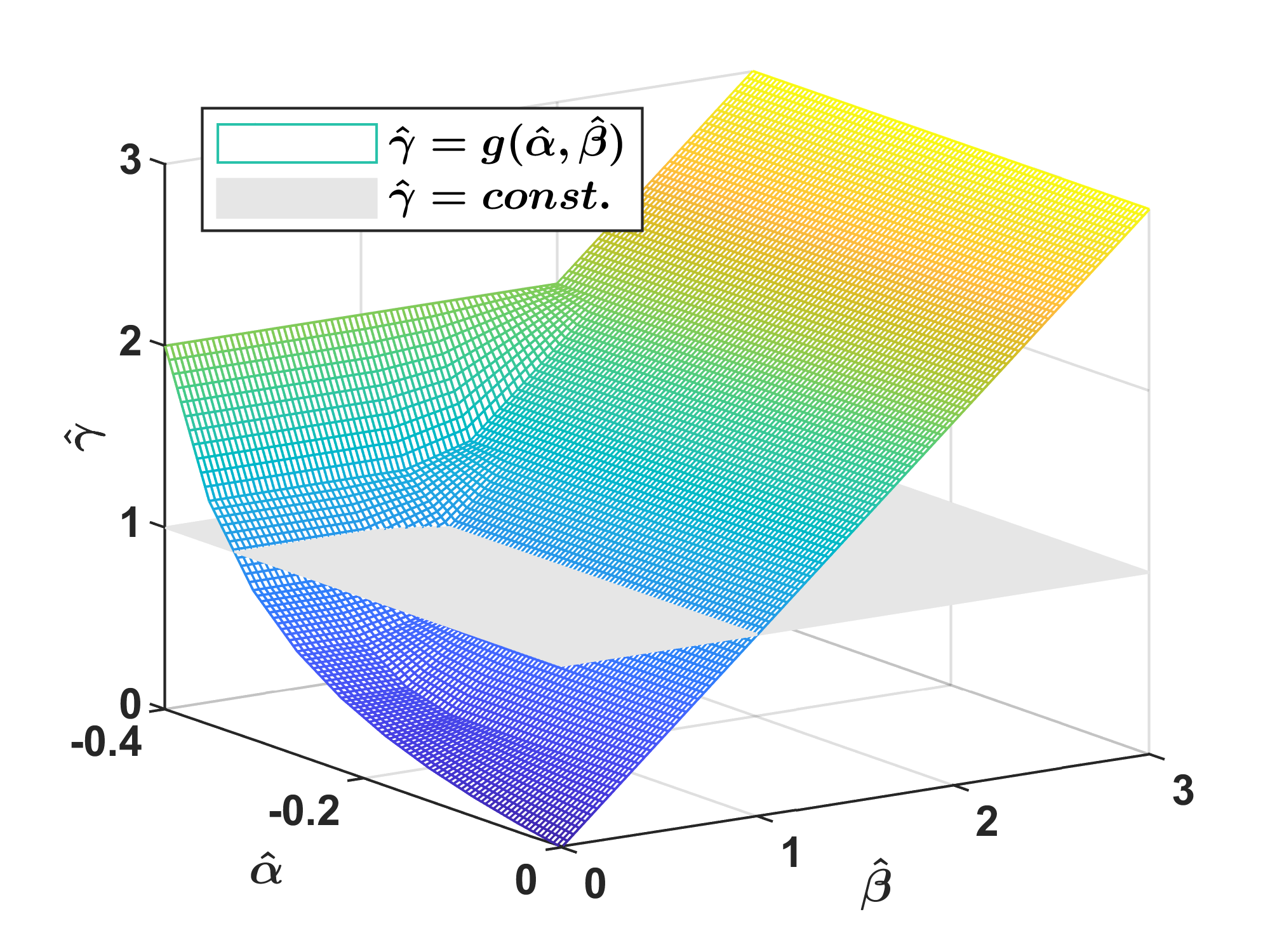}
    \caption{$g(\hat{\alpha}.\hat{\beta})$ is the function of sector bounds $[\hat{\alpha}.\hat{\beta}]$.}
    \label{fig:Current_block_gain}
\end{figure}

\begin{proposition}
Given the class $\Sigma$ buck converter modeled in \cite{Cui2018a}, the $\mathcal{L}_2$ gain from the sampled output voltage sequence $\{\tilde{v}[n]\}$ to one-cycle-delayed inductor current sequence $\{\tilde{i}^{p}_v[n]\}$ is bounded from above by
\begin{align}\label{egn:gamma_vi_buck}
    \Gamma_{v \to i} \le \frac{T_s^{\text{ss}}}{L} g(\hat{\alpha}, \hat{\beta}),
\end{align}
where $g(\hat{\alpha}, \hat{\beta})$ follows Fig.\,\ref{fig:Current_block_gain}, $T_s^{\text{ss}}$ is the steady-state switching period, 
\begin{align}
    T_s^{\text{ss}} & \triangleq T_{\text{on}} +  T_{\text{off}}.
\end{align}
\end{proposition}


\begin{proposition}
Given the class $\Sigma$ buck converter using constant on-time modeled in \cite{Cui2018a}, the $\mathcal{L}_2$ gain from the one-cycle-delayed inductor current sequence $\{\tilde{i}^{p}_v[n]\}$ to sampled output voltage sequence $\{\tilde{v}[n]\}$ is bounded from above by
\begin{align}
    \Gamma_{i \to v} &  \le \frac{R}{\left( 1 + \frac{T_{\text{on}}}{2\tau_2}\right)} \frac{T_s^{\text{max}}}{T_s^{\text{min}}},
\end{align}
where $T_s^{\text{min}}$ and $T_s^{\text{max}}$ are the shortest switching period and longest switching period, respectively, and $\tau_2$ is the $L/R$ time constant,
\begin{align}
    T_s^{\text{max}} & \triangleq T_{\text{on}} +  T_{\text{off}}^{\text{max}}, \\
    T_s^{\text{min}} & \triangleq T_{\text{on}} +  T_{\text{off}}^{\text{min}},\\
    \tau_2 &\triangleq \frac{L}{R}.
\end{align}
\end{proposition}
\begin{proposition}
The current control loop of the class $\Sigma$ buck converter using constant on-time current-mode control is globally asymptotically stable if 
\begin{align}
    & g(\hat{\alpha}, \hat{\beta}) < \left(\tau_2 + \frac{T_{\text{on}}}{2}\right)\frac{T_s^{\text{min}}}{T_s^{\text{max}}} \frac{1}{T_s^{ss}}\:\: \text{and}\\
   & T_s^{\text{max}}\left( 1 + \frac{T_{\text{on}}}{2\tau_2}\right) < \tau_1,
\end{align}
where $g(\hat{\alpha}, \hat{\beta})$ follows Fig.\,\ref{fig:Current_block_gain}, $T_s^{\text{min}}$ and $T_s^{\text{max}}$ are the shortest switching period and longest switching period, respectively,
$\tau_1$ is the $RC$ time constant, and $\tau_2$ is the $L/R$ time constant,
\begin{align}
    T_s^{\text{max}} & \triangleq T_{\text{on}} +  T_{\text{off}}^{\text{max}}, \\
    T_s^{\text{min}} & \triangleq T_{\text{on}} +  T_{\text{off}}^{\text{min}},\\
    \tau_1 &\triangleq RC, \\
    \tau_2 &\triangleq \frac{L}{R}.
\end{align}
\end{proposition}

\subsection{Current-Mode Boost Converter Using Constant Off-Time}
Consider a class $\Sigma$ boost converter using constant off-time defined in \cite{Cui2018a}, the on-time in steady state is denoted by $T_{\text{off}}$.
The time-varying on-time is bounded from above by
\begin{align} \label{eqn:toff_bd_boost}
    T^{\text{min}}_{\text{on}} \le t_{\text{on}}[n] \le T^{\text{max}}_{\text{on}}.
\end{align}
The large-signal stability can be described by the following three propositions:
\begin{proposition}
Given the class $\Sigma$ boost converter using constant off-time control, the $\mathcal{L}_2$ gain from the sampled output voltage sequence $\{\tilde{v}[n]\}$ to one-cycle-delayed inductor current sequence $\{\tilde{i}^{p}_p[n]\}$ is bounded from above by
\begin{align}\label{egn:gamma_vi_boost}
    \Gamma_{v \to i} \le \frac{T_{\text{off}}}{L} g(\hat{\alpha}, \hat{\beta}),
\end{align}
where $g(\hat{\alpha}, \hat{\beta})$ follows Fig.\,\ref{fig:Current_block_gain}, $T_s^{\text{ss}}$ is the steady-state switching period, 
\begin{align}
    T_s^{\text{ss}} & \triangleq T_{\text{on}} +  T_{\text{off}}.
\end{align}
\end{proposition}

\begin{proposition}
Given the class $\Sigma$ boost converter using constant off-time control, the $\mathcal{L}_2$ gain from the one-cycle-delayed inductor current sequence $\{\tilde{i}^{p}_p[n]\}$ to sampled output voltage sequence $\{\tilde{v}[n]\}$ is bounded from above by the following:

\vspace{+8pt}
\noindent (i) if $\left((1-\lambda ) T_{\text{off}} + \frac{V_{\text{out}}L}{V_{\text{in}}R} \right)\left( 1-\frac{T_s^{\text{ss}}}{RC} -\frac{T_{s}^{\text{max}}}{RC}-\frac{ T^2_{\text{off}}}{2LC}\right) + \left( \lambda T_{\text{off}} - \frac{V_{\text{out}}L}{V_{\text{in}}R} \right) \ge 0 $\:\:and
\begin{align} \label{eqn:ivgain_case1_cotcm}
    \Gamma_{i \to v} &  \le \frac{R}{\left( \frac{T_s^{\text{ss}} + T_s^{\text{min}}}{T_{\text{off}}} + \frac{T_{\text{off}}}{2\tau_2}\right)},
\end{align}

\noindent or (ii) if $\left((1-\lambda ) T_{\text{off}} + \frac{V_{\text{out}}L}{V_{\text{in}}R} \right)\left( 1-\frac{T_s^{\text{ss}}}{RC} -\frac{T_{s}^{\text{max}}}{RC}-\frac{ T^2_{\text{off}}}{2LC}\right) + \left( \lambda T_{\text{off}} - \frac{V_{\text{out}}L}{V_{\text{in}}R} \right) < 0 $\:\:and
\begin{align} \label{eqn:ivgain_case2_cotcm}
    \Gamma_{i \to v} &  \le \frac{2\tau_2\left( T_s^{\text{max}} +  T_s^{\text{ss}}\right) + T_{\text{off}}^2}{2\tau_2\left( T_s^{\text{min}} +  T_s^{\text{ss}}\right) + T_{\text{off}}^2} \frac{2\tau_2\frac{V_{\text{out}}}{V_{\text{in}}}+(1-2\lambda)T_{\text{off}}}{2\tau_1-T_s^{\text{ss}}-T_s^{\text{max}}-\frac{T_{\text{off}}^2}{2\tau_2}} R,
\end{align}
where $T_s^{\text{min}}$ and $T_s^{\text{max}}$ are the shortest switching period and longest switching period, respectively, $\tau_1$ is the $RC$ time constant, and $\tau_2$ is the $L/R$ time constant,
\begin{align}
    T_s^{\text{max}} & \triangleq T_{\text{off}} +  T_{\text{on}}^{\text{max}}, \\
    T_s^{\text{min}} & \triangleq T_{\text{off}} +  T_{\text{on}}^{\text{min}}, \\
    \tau_1 &\triangleq RC, \\
    \tau_2 &\triangleq \frac{L}{R}.
\end{align}
\end{proposition}

\begin{proposition}\label{theore: boost_offtime_stability_cotcmboost}
The current control loop of the class $\Sigma$ boost converter using constant off-time control is globally asymptotically stable

\vspace{+8pt}
\noindent (i) if $\left((1-\lambda ) T_{\text{off}} + \frac{V_{\text{out}}L}{V_{\text{in}}R} \right)\left( 1-\frac{T_s^{\text{ss}}}{RC} -\frac{T_{s}^{\text{max}}}{RC}-\frac{ T^2_{\text{off}}}{2LC}\right) + \left( \lambda T_{\text{off}} - \frac{V_{\text{out}}L}{V_{\text{in}}R} \right) \ge 0 $\:\:and
\begin{align}
    g(\hat{\alpha}, \hat{\beta}) &  \le 
    \frac{1}{2} + \tau_2
   \left(\frac{T_s^{\text{ss}}+T_s^{\text{min}}}{T_s^{\text{ss}}T_{\text{off}}} \right);
\end{align}

\noindent or (ii) if $\left((1-\lambda ) T_{\text{off}} + \frac{V_{\text{out}}L}{V_{\text{in}}R} \right)\left( 1-\frac{T_s^{\text{ss}}}{RC} -\frac{T_{s}^{\text{max}}}{RC}-\frac{ T^2_{\text{off}}}{2LC}\right) + \left( \lambda T_{\text{off}} - \frac{V_{\text{out}}L}{V_{\text{in}}R} \right) < 0 $\:\:and
\begin{align}
    &g(\hat{\alpha}, \hat{\beta}) \le\\\nonumber
    &\frac{2\tau_2\left( T_s^{\text{min}} +  T_s^{\text{ss}}\right) + T_{\text{off}}^2}{2\tau_2\left( T_s^{\text{max}} +  T_s^{\text{ss}}\right) + T_{\text{off}}^2} \frac{2\tau_1-T_s^{\text{ss}}-T_s^{\text{max}}-\frac{T_{\text{off}}^2}{2\tau_2}}{2\frac{V_{\text{out}}}{V_{\text{in}}}+(1-2\lambda)\frac{T_{\text{off}}}{\tau_2}}T_{\text{off}};
\end{align}
where $g(\hat{\alpha}, \hat{\beta})$ follows Fig.\,\ref{fig:Current_block_gain}, $T_s^{\text{min}}$ and $T_s^{\text{max}}$ are the shortest switching period and longest switching period, respectively, $\tau_1$ is the $RC$ time constant, and $\tau_2$ is the $L/R$ time constant,
\begin{align}
    T_s^{\text{max}} & \triangleq T_{\text{on}} +  T_{\text{off}}^{\text{max}}, \\
    T_s^{\text{min}} & \triangleq T_{\text{on}} +  T_{\text{off}}^{\text{min}},\\
    \tau_1 &\triangleq RC, \\
    \tau_2 &\triangleq \frac{L}{R}.
\end{align}
\end{proposition}

\section{Design Equations for Three Other Typical Current-Mode Converters} \label{ref:sec_design_eqn}
We directly use the settling $N_w$ as the metric for settling cycles,
\begin{align} \label{eqn:settlecycle_repeat}
    N_w \triangleq \text{max} \bigg\{\bigg|\frac{4}{\text{ln}(|a_{\text{min}}|)}\bigg|,\bigg|\frac{4}{\text{ln}(|a_{\text{max}}|)}\bigg|\bigg\}.
\end{align}
The worst-case overshoot $O_w$ follows
\begin{align} \label{eqn:os1p1z_repeat}
    O_w \triangleq \text{max}\bigg\{\frac{b-a_{\text{min}}}{1-b}\,,\,0\bigg\}.
\end{align}
The design equations of the constant on-time current-mode converter, constant off-time current-mode converter, fixed-frequency peak current-mode converter and fixed-frequency valley current-mode converter are illustrated in Tables\,\ref{table:cont_off_cm_theory},
\ref{table:cofft_on_cm_theory}, 
\ref{table:ff_peak_cm_theory}, and \ref{table:ff_valley_cm_theory}.
\begin{table*}[tb]
    \caption{Design Theory and Equations for Constant Off-Time Current-Mode Converters}
    \label{table:cont_off_cm_theory}
    \centering
    \begin{tabular}
    {|m{2.4in}|m{1.025in}|m{1.025in}|m{0.55in}|}
        \hline
        \textbf{Stability Criteria} & $\mathbf{a_{\text{min}}}$ & $\mathbf{a_{\text{max}}}$ & \textbf{b} \\
        \hline
        \vspace{3pt}
        $\Lambda_{ub} \le \frac{m_1}{2}$
        \vspace{2pt}
        & $1 - \frac{m_1}{(m_1 - \Lambda_{ub})}$
        \vspace{2pt}
        & $1 - \frac{m_1}{(m_1 + \Lambda_{ub})}$
        \vspace{2pt}
        & 0 \\
        \hline
    \end{tabular}
\end{table*}

\begin{table*}[tb]
    \caption{Design Theory and Equations for the Constant On-Time Current-Mode Converters}
    \label{table:cofft_on_cm_theory}
    \centering
    \begin{tabular}
    {|m{2.4in}|m{1.025in}|m{1.025in}|m{0.55in}|}
        \hline
        \textbf{Stability Criteria}& $\mathbf{a_{\text{min}}}$ & $\mathbf{a_{\text{max}}}$ & \textbf{b} \\
        \hline
        \vspace{3pt}
        $\Lambda_{ub} \le \frac{m_2}{2}$
         \vspace{2pt}
        & $1 - \frac{m_2}{(m_2 - \Lambda_{ub})}$
         \vspace{2pt}
        & $1 - \frac{m_2}{(m_2 + \Lambda_{ub})}$
         \vspace{2pt}
        & 0 \\
         \hline
    \end{tabular}
\end{table*}

\begin{table*}[tb]
    \caption{Design Theory and Equations for the Fixed-Frequency Peak Current-Mode Converters}
    \label{table:ff_peak_cm_theory}
    \centering
    \begin{tabular}
    {|m{2.4in}|m{1.025in}|m{1.025in}|m{0.55in}|}
        \hline
        \textbf{Stability Criteria}& $\mathbf{a_{\text{min}}}$ & $\mathbf{a_{\text{max}}}$ & \textbf{b} \\
        \hline
        \vspace{3pt}
        $\Lambda_{ub} \le \frac{m_1 - m_2}{2}$ 
        \vspace{2pt}
        & $\frac{-\Lambda_{ub}-m_2}{(m_1 - \Lambda_{ub})}$
        \vspace{2pt}
        & $\frac{\Lambda_{ub} - m_2}{(m_1 + \Lambda_{ub})}$
        \vspace{2pt}
        & $-\frac{m_2}{m_1}$
        \vspace{2pt} \\
         \hline
    \end{tabular}
\end{table*}

\begin{table*}[tb]
    \caption{Design Theory and Equations for the Fixed-Frequency Valley Current-Mode Converters}
    \label{table:ff_valley_cm_theory}
    \centering
    \begin{tabular}
    {|m{2.4in}|m{1.025in}|m{1.025in}|m{0.55in}|}
        \hline
        \textbf{Stability Criteria}& $\mathbf{a_{\text{min}}}$ & $\mathbf{a_{\text{max}}}$ & \textbf{b} \\
        \hline
        \vspace{3pt}
        $\Lambda_{ub} \le \frac{m_2 - m_1}{2}$ 
        \vspace{2pt}
        & $\frac{-\Lambda_{ub}-m_1}{(m_2 - \Lambda_{ub})}$
        \vspace{2pt}
        & $\frac{\Lambda_{ub} - m_1}{(m_2 + \Lambda_{ub})}$
        \vspace{2pt}
        & $-\frac{m_1}{m_2}$ 
        \vspace{2pt}
        \\
         \hline
    \end{tabular}
\end{table*}

\section{Sector Boundedness of Static Mapping $\mathcal{T}$} \label{ref:sec_boun_Tmap}
\begin{proof}
We denote the bandwidth limit of interference $w(t)$ by $f_{ub}$ and amplitude limit of the interference by $A_{ub}$.
From Bernstein's inequality \cite{Lapidoth2009},
\begin{align}
    \Bigg | \frac{\text{d}\,w(t)}{\text{d}\,t} \Bigg |\le 4\pi f_{ub} A_{ub} \triangleq \Lambda_{ub}, \: \forall t \in \mathbb{R}.
\end{align}
We linearly transform the static mapping $\mathcal{T}$ to the origin, which results in $\tilde{\mathcal{T}}$
\begin{align}
    i_c  = I_c & + \tilde{i}_c,   \\
    i_p  = I_p & + \tilde{i}_p,   \\
    \mathcal{T}:i_c & \rightarrow i_p  \label{eqn:defT_2},\\
    \tilde{\mathcal{T}}: \tilde{i}_c & \rightarrow \tilde{i}_p \label{eqn:deftildeT_2}.
\end{align}
The relationship between $\tilde{\mathcal{T}}$ and $\psi$ follows
\begin{align}
        \frac{\text{d}\,(\tilde{\mathcal{T}}^{-1})}{\text{d}\,x}  & = 1 + G_0\psi^{'} \label{eqn:devTinvx_1},\:
        \text{where}\; \psi^{'} = \frac{\text{d}\,\psi(x)}{\text{d}\,x},\\
       \frac{\text{d}\,\tilde{\mathcal{T}}}{\text{d}\,x} & = \frac{1}{1 + G_0\psi^{'}}\label{eqn:devTx_1},
\end{align}
\begin{align}
         \psi^{'} & = \frac{1}{G_0} \left( \frac{1}{\frac{\text{d}\,\tilde{\mathcal{T}}}{\text{d}\,x}} - 1  \right) = \frac{1}{G_0} \left( \frac{1}{\frac{\text{d}\,\mathcal{T}}{\text{d}\,x}} - 1  \right) \label{eqn:devpsi_1}.
\end{align}
The relationship between $\psi(t)$ and $w(t)$ follows
\begin{align}
     \psi(t) & = w(t+T_{\text{on}}) - w(T_{\text{on}}) \label{eqn:psit_wt_1},\\
     \frac{\text{d}\,\psi(t)}{\text{d}\,t} & = \frac{\text{d}\,w(t)}{\text{d}\,t}, \\
     -\Lambda_{ub} \le  & \frac{\text{d}\,w(t)}{\text{d}\,t} \le  \Lambda_{ub} \label{eqn:dpsit_dwt_1}.
\end{align}
From (\ref{eqn:dpsit_dwt_1}) and (\ref{eqn:devTx_1})
\begin{align}
         \frac{1}{1 + \Lambda_{ub} G_0} \le \frac{\text{d}\,\tilde{\mathcal{T}}}{\text{d}\,x} \le \frac{1}{1 - \Lambda_{ub} G_0}   \label{eqn:devTx_1_1},
\end{align}
\begin{align}
     \tilde{\mathcal{T}}(x) = \tilde{\mathcal{T}}(x) -  \tilde{\mathcal{T}}(0) = \int_0^{x} \frac{\text{d} \tilde{\mathcal{T}}}{\text{d} x}\,\text{d} x. \label{eqn:boundTtilde}
\end{align}
Therefore, we have proven $\tilde{\mathcal{T}}$ is sector-bounded
\begin{align}
    \frac{x}{1 + \Lambda_{ub} G_0} \le \tilde{\mathcal{T}}(x) \le \frac{x}{1 - \Lambda_{ub} G_0}.
\end{align}
We have proven $\mathcal{T}$ is $\alpha$ sector-bounded with
\begin{align}
    \alpha &= I_p, \\
    K_{lb} &= \frac{1}{1 + \Lambda_{ub} G_0}, \quad 
    K_{ub} = \frac{1}{1 - \Lambda_{ub} G_0}.
\end{align}
\end{proof}

\section{Proof of Theorem \ref{th:equexcond}}  \label{proof:equexcond}
\begin{proof}
If first-event-triggering with latching is used, $\mathcal{T}$ is a monotonically increasing mapping. Therefore, the proposition that $\mathcal{T}$ is continuous is equivalent to the proposition that $\mathcal{T}$ is onto. We first prove that if the current sensor output \mbox{$m_{1}t+w(t)$} is strictly monotonic, $\mathcal{T}$ is an onto mapping. We next prove that if $\mathcal{T}$ is an onto mapping, the current sensor output $m_{1}t+w(t)$ is strictly monotonic. The contradiction method is used in the proof.

1. The current sensor output $m_{1}t+w(t)$ is strictly \\
monotonic $\Longrightarrow$ $\mathcal{T}$ is an onto function.\\
Given any $I_p$, there uniquely exists an on-time $t_{\text{on}}$ which satisfies
\begin{align}
  I_{p}= I_{v} + m_{1}t_{\text{on}}.
\end{align}
Because the measured inductor current waveform $m_{1}t+w(t)$ is strictly monotonic, given any $t_{\text{on}}$, there uniquely exists a current command $I_c$ which achieves \,$t_{\text{on}}$ by following the first-event-triggering with latching mechanism
\begin{align}
I_c = I_v+m_{1}t_{\text{on}}+w(t_{\text{on}}),
\end{align}
where $I_{v}$ is the valley current. To consolidate, given any $I_p$, there uniquely exists a $t_{\text{on}}$ and given any $t_{\text{on}}$, there uniquely exists an $I_c$.
Therefore, we have proven that $\mathcal{T}$ is an onto function.

2. $\mathcal{T}$ is an onto function $\Longrightarrow$  the measured inductor current waveform $ m_{1}t+w(t)$ is strictly monotonic.\\
We prove it by contradiction. We assume that $m_{1}t+w(t)$ is not strictly monotonic, hence there exist $t_{\text{on1}}$ and $t_{\text{on2}}$ such that
\begin{equation}
\begin{aligned}
t_{\text{on2}} >  t_{\text{on1}} \Longrightarrow m_{1}t_{\text{on2}}+w(t_{\text{on2}}) \leq m_{1}t_{\text{on1}}+w(t_{\text{on1}}).
\end{aligned}
\end{equation}
The corresponding $I_{p1}$ of $t_{\text{on1}}$ and $I_{p2}$ of $t_{\text{on2}}$ are 
\begin{equation}
\begin{aligned}
I_{p1}= I_{v} + m_{1}t_{\text{on1}}, \quad  \quad I_{p2}= I_{v} + m_{1}t_{\text{on2}}. 
\end{aligned}
\end{equation}
Because $\mathcal{T}$ is an onto function,
there exist $ I_{c1},I_{c2}$ such that
\begin{align}
    \mathcal{T}(I_{c1})= I_{p1}, \quad \quad
    \mathcal{T}(I_{c2})= I_{p2}.
\end{align}
From the definition of $I_{c1}$ and $I_{c2}$,
\begin{align}
    I_{c1} &= I_{v} + m_{1}t_{\text{on1}}+w(t_{\text{on1}}), \nonumber \\
    I_{c2} &= I_{v} + m_{1}t_{\text{on2}}+w(t_{\text{on2}}) \label{eqn:icexp}, 
\end{align}
where $I_{v}$ is the valley current. \\
Because of the first-event-trigger with latching mechanism, for $I_{c2} > I_{c1}$, we have  $t_{\text{on2}} > t_{\text{on1}}$. From (\ref{eqn:icexp}), we obtained
\begin{equation}
\begin{aligned}
 m_{1}t_{\text{on2}}+w(t_{\text{on2}}) > m_{1}t_{\text{on1}}+w(t_{\text{on1}}), \\
\end{aligned}
\end{equation}
which contradicts the assumption that $m_{1}t_{\text{on1}}+w(t_{\text{on1}})$ is not strictly monotonic. Hence we have proven that if $\mathcal{T}$ is an onto function, the measured inductor current waveform $m_{1}t+w(t)$ is strictly monotonic.
\end{proof}

\end{appendices}
\newpage
{\setstretch{1}\vspace{\baselineskip}
\bibliographystyle{ieeetr}
\bibliography{main.bib}

\begin{thebibliography}{10}

\bibitem{Fernandes2016}
R.~Fernandes and O.~Trescases, ``A multimode {1-MHz PFC} front end with digital
  peak current modulation,'' {\em IEEE Transactions on Power Electronics},
  vol.~31, pp.~5694--5708, Aug. 2016.

\bibitem{Lee2013}
F.~C. Lee and Q.~Li, ``High-frequency integrated point-of-load converters :
  Overview,'' vol.~28, no.~9, pp.~4127--4136, 2013.

\bibitem{Corradini2015}
L.~Corradini, D.~Maksimovi{\'{c}}, P.~Mattavelli, and R.~Zane, {\em Digital
  Control of High-Frequency Switched-Mode Power Converters}, vol.~48.
\newblock John Wiley and Sons, 2015.

\bibitem{erickson2007}
R.~W. Erickson and D.~Maksimovic, {\em Fundamentals of Power Electronics}.
\newblock Springer Science and Business Media, 2007.

\bibitem{Bao2021}
C.~Bao, S.~Guptam, and S.~K. Mazumder, ``Modeling and analysis of
  peak-current-controlled differential mode {\'{c}}uk inverter,'' {\em
  Proceedings of the 2021 IEEE 12th International Symposium on Power
  Electronics for Distributed Generation Systems, PEDG 2021}, June 2021.

\bibitem{Cui2018a}
X.~Cui and A.-T. Avestruz, ``A new framework for cycle-by-cycle digital control
  of megahertz-range variable frequency buck converters,'' in {\em 2018 IEEE
  19th Workshop on Control and Modeling for Power Electronics (COMPEL)},
  (Padova), pp.~1--8, 2018.

\bibitem{Cui2021apec}
X.~Cui, C.~Deng, and A.-T. Avestruz, ``A fast response dc-dc converter with
  programmable ripple for combined distributed computation and communication,''
  in {\em 2021 IEEE Applied Power Electronics Conference and Exposition
  (APEC)}, pp.~468--473, 2021.

\bibitem{Cui2019c}
X.~Cui, C.~Keller, and A.-T. Avestruz, ``Cycle-by-cycle digital control of a
  multi-megahertz variable-frequency boost converter for automatic power
  control of lidar,'' in {\em 2019 IEEE Energy Conversion Congress and
  Exposition (ECCE)}, (Baltimore), pp.~702--711, 2019.

\bibitem{xinbo2020}
Y.~Wang, X.~Ruan, Q.~Jin, H.~Xi, X.~Xiong, Y.~Leng, and Y.~Li, ``Elimination of
  the interaction of the converters in switch-linear hybrid envelope tracking
  power supplies,'' {\em IEEE Transactions on Power Electronics}, vol.~35,
  pp.~2053--2066, Feb. 2020.

\bibitem{Lazarevic2019}
V.~Z. Lazarevic, I.~Zubitur, M.~Vasic, J.~A. Oliver, P.~Alou, G.~Patchin,
  J.~Eltze, and J.~A. Cobos, ``High-efficiency high-bandwidth four-quadrant
  fully digitally controlled gan-based tracking power supply system for linear
  power amplifiers,'' {\em IEEE Journal of Emerging and Selected Topics in
  Power Electronics}, vol.~7, pp.~664--678, June 2019.

\bibitem{Wang2018}
H.~Wang, M.~Han, R.~Han, J.~M. Guerrero, and J.~C. Vasquez, ``A decentralized
  current-sharing controller endows fast transient response to parallel dc-dc
  converters,'' {\em IEEE Transactions on Power Electronics}, vol.~33,
  pp.~4362--4372, May 2018.

\bibitem{Chang2020}
F.~Chang, X.~Cui, M.~Wang, W.~Su, and A.~Q. Huang, ``Large-signal stability
  criteria in dc power grids with distributed-controlled converters and
  constant power loads,'' {\em IEEE Transactions on Smart Grid}, vol.~11,
  no.~6, pp.~5273--5287, 2020.

\bibitem{Svikovic2015}
V.~Svikovic, J.~J. Cortes, P.~Alou, J.~A. Oliver, O.~Garcia, and J.~A. Cobos,
  ``Multiphase current-controlled buck converter with energy recycling output
  impedance correction circuit (oicc),'' {\em IEEE Transactions on Power
  Electronics}, vol.~30, pp.~5207--5222, Sept. 2015.

\bibitem{Huang2016c}
W.~Huang and B.~Lehman, ``A compact coupled inductor for interleaved multiphase
  dc-dc converters,'' {\em IEEE Transactions on Power Electronics}, vol.~31,
  pp.~6770--6775, Oct. 2016.

\bibitem{Halivni2020}
B.~Halivni and M.~M. Peretz, ``Digital controller for high-performance
  multiphase vrm with current balancing and near-ideal transient response,''
  {\em Conference Proceedings - IEEE Applied Power Electronics Conference and
  Exposition - APEC}, vol.~2020-March, pp.~2206--2213, Mar. 2020.

\bibitem{Karanayil2017}
B.~Karanayil, M.~Ciobotaru, and V.~G. Agelidis, ``Power flow management of
  isolated multiport converter for more electric aircraft,'' {\em IEEE
  Transactions on Power Electronics}, vol.~32, pp.~5850--5861, July 2017.

\bibitem{Ding2020}
L.~Ding, S.-C. Wong, and C.~K. Tse, ``Bifurcation analysis of a current-mode
  controlled dc cascaded system and applications to design,'' {\em IEEE Journal
  of Emerging and Selected Topics in Power Electronics}, vol.~8,
  pp.~3214--3224, Dec. 2020.

\bibitem{Li2012a}
Y.~Li, K.~R. Vannorsdel, A.~J. Zirger, M.~Norris, and D.~Maksimovic, ``Current
  mode control for boost converters with constant power loads,'' {\em IEEE
  Transactions on Circuits and Systems I: Regular Papers}, vol.~59, no.~1,
  pp.~198--206, 2012.

\bibitem{Roy2020}
R.~Roy and S.~Kapat, ``Discrete-time framework for analysis and design of
  digitally current mode controlled intermediate bus architectures for fast
  transient and stability,'' {\em IEEE Journal of Emerging and Selected Topics
  in Power Electronics}, vol.~8, pp.~3237--3249, Dec. 2020.

\bibitem{Abedi2016}
M.~R. Abedi and K.~Y. Lee, ``Disturbance rejection of peak current-mode control
  for bidirectional battery charging,'' in {\em 2016 IEEE Power and Energy
  Society General Meeting (PESGM)}, pp.~1--5, IEEE Computer Society, Nov. 2016.

\bibitem{alireza2019}
A.~Ramyar, X.~Cui, and A.~T. Avestruz, ``Two-port up/down dc-dc converter for
  two-dimensional maximum power point tracking of differential diffusion charge
  redistribution solar panel,'' in {\em 2019 IEEE 20th Workshop on Control and
  Modeling for Power Electronics (COMPEL)}, pp.~1--8, IEEE, June 2019.

\bibitem{Ziegler2009}
S.~Ziegler, R.~C. Woodward, H.~H.~C. Iu, and L.~J. Borle, ``Current sensing
  techniques: A review,'' {\em IEEE Sensors Journal}, vol.~9, pp.~354--376,
  Apr. 2009.

\bibitem{taeed2014compel}
F.~Taeed and M.~Nymand, ``A novel high performance and robust digital peak
  current mode controller for dc-dc converters in ccm,'' in {\em 2014 IEEE 15th
  Workshop on Control and Modeling for Power Electronics (COMPEL)}, pp.~1--5,
  IEEE, June 2014.

\bibitem{Chen2003tpel}
J.~Chen, A.~Prodi{\'{c}}, R.~W. Erickson, and D.~Maksimovi{\'{c}}, ``Predictive
  digital current programmed control,'' {\em IEEE Transactions on Power
  Electronics}, vol.~18, pp.~411--419, Jan. 2003.

\bibitem{Femia2020}
N.~Femia, K.~Stoyka, and G.~{Di Capua}, ``Impact of inductors saturation on
  peak-current mode control operation,'' {\em IEEE Transactions on Power
  Electronics}, vol.~35, pp.~10969--10981, Oct. 2020.

\bibitem{chen2020}
Y.~Chen, F.~Xie, B.~Zhang, D.~Qiu, and M.~Xu, ``Analysis of digital
  pcm-controlled boost converter with trailing-edge modulation based on
  z-domain and describing-function model,'' {\em IEEE Journal of Emerging and
  Selected Topics in Power Electronics}, vol.~8, pp.~3250--3259, Dec. 2020.

\bibitem{DiCapua2016}
G.~{Di Capua} and N.~Femia, ``A novel method to predict the real operation of
  ferrite inductors with moderate saturation in switching power supply
  applications,'' {\em IEEE Transactions on Power Electronics}, vol.~31, no.~3,
  pp.~2456--2464, 2016.

\bibitem{Redl1981a}
R.~Redl and I.~Novak, ``Instabilities in current-mode controlled switching
  voltage regulators,'' in {\em 1981 IEEE Annual Power Electronics Specialists
  Conference}, pp.~17--28, 1981.

\bibitem{Li2016b}
Q.~Li, J.~Zhang, and U.~Epple, ``Design and exit chart analysis of a doubly
  iterative receiver for mitigating impulsive interference in ofdm systems,''
  {\em IEEE Transactions on Communications}, vol.~64, pp.~1726--1738, Apr.
  2016.

\bibitem{Chen2015a}
Y.~Chen, Y.~R. Nan, Q.~G. Kong, and S.~H. Zhong, ``An input-adaptive
  self-oscillating boost converter for fault-tolerant led driving with
  wide-range ultralow voltage input,'' {\em IEEE Transactions on Power
  Electronics}, vol.~30, pp.~2743--2752, May 2015.

\bibitem{Simon-Muela2008}
A.~Sim{\'{o}}n-Muela, S.~Petibon, C.~Alonso, and J.~L. Chaptal, ``Practical
  implementation of a high-frequency current-sense technique for vrm,'' {\em
  IEEE Transactions on Industrial Electronics}, vol.~55, no.~9, pp.~3221--3230,
  2008.

\bibitem{Chen2014d}
W.~W. Chen, J.~F. Chen, T.~J. Liang, L.~C. Wei, and W.~Y. Ting, ``Designing a
  dynamic ramp with an invariant inductor in current-mode control for an
  on-chip buck converter,'' {\em IEEE Transactions on Power Electronics},
  vol.~29, no.~2, pp.~750--758, 2014.

\bibitem{kuo1987automatic}
B.~C. Kuo, {\em Automatic Control Systems}.
\newblock Prentice Hall PTR, 1987.

\bibitem{Liberzon2003a}
D.~Liberzon, {\em Switching in Systems and Control}.
\newblock Systems and Control: Foundations and Applications, Boston, MA:
  Birkh{\"{a}}user Boston, 2003.

\bibitem{McNeill2004}
N.~McNeill, N.~K. Gupta, and W.~G. Armstrong, ``Active current transformer
  circuits for low distortion sensing in switched mode power converters,'' {\em
  IEEE Transactions on Power Electronics}, vol.~19, no.~4, pp.~908--917, 2004.

\bibitem{Singh2008a}
R.~P. Singh and A.~M. Khambadkone, ``Giant magneto resistive (gmr) effect based
  current sensing technique for low voltage/high current voltage regulator
  modules,'' {\em IEEE Transactions on Power Electronics}, vol.~23,
  pp.~915--925, Mar. 2008.

\bibitem{lee2003design}
T.~H. Lee, {\em The Design of CMOS Radio-Frequency Integrated Circuits}.
\newblock Cambridge University Press, 2003.

\bibitem{Jury1964}
E.~Jury and B.~Lee, ``On the stability of a certain class of nonlinear
  sampled-data systems,'' {\em IEEE Transactions on Automatic Control}, vol.~9,
  no.~1, pp.~51--61, 1964.

\bibitem{Brockett1966}
R.~W. Brockett, ``The status of stability theory for deterministic systems,''
  {\em IEEE Transactions on Automatic Control}, vol.~11, pp.~596--606, July
  1966.

\bibitem{khalil2002nonlinear}
H.~K. Khalil, {\em Nonlinear Systems}.
\newblock Upper Saddle River, 2002.

\bibitem{Larsen2001}
M.~Larsen and P.~V. Kokotovi, ``A brief look at the tsypkin criterion: from
  analysis to design,'' {\em International Journal of Adaptive Control and
  Signal Processing}, vol.~15, pp.~121--128, Mar. 2001.

\bibitem{Okuyama2014}
Y.~Okuyama, {\em Discrete Control Systems}.
\newblock Springer, Nov. 2014.

\bibitem{xiaofanacc2019}
X.~Cui and A.-T. Avestruz, ``Switching-synchronized sampled-state space
  modeling and digital controller for a constant off-time, current-mode boost
  converter,'' in {\em 2019 American Control Conference (ACC)}, (Philadelphia),
  pp.~1--8, 2019.

\bibitem{Avestruz2022}
X.~Cui and A.-T. Avestruz, ``Large-signal stability analysis of current-mode
  dc-dc converters https://tinyurl.com/3fcs3aaw.'' 2022.

\bibitem{Lapidoth2009}
A.~Lapidoth, {\em A Foundation in Digital Communication}.
\newblock Cambridge University Press, 2009.

\end{thebibliography}
}

\end{document}